\documentclass[12pt]{article}

\usepackage{algorithm}
\usepackage{algpseudocode}
\usepackage{amsbsy, amsfonts, amsmath, amssymb, amsthm}

\newtheorem{theorem}{Theorem}
\newtheorem{lemma}[theorem]{Lemma}

\usepackage{dsfont}

\usepackage{bm}

\usepackage{booktabs, threeparttable, tabulary}
\usepackage{setspace}

\usepackage{natbib}

\usepackage{color}

\usepackage{graphicx}

\usepackage{physics} 

\usepackage{mathrsfs}

\allowdisplaybreaks

\newcommand{\E}{\mathbb{E}\,}

\newcommand{\bb}{\bm{\beta}}
\newcommand{\bg}{\bm{\gamma}}
\newcommand{\bth}{\bm{\theta}}
\newcommand{\bx}{\bm{x}}
\newcommand{\ox}{\overline{x}}

\newcommand{\rom}[1]{\uppercase\expandafter{\romannumeral #1\relax}}
\newcommand{\dif}{\mathrm{d}}

\newcommand{\pkg}[1]{{\normalfont\fontseries{b}\selectfont #1}}
\let\proglang=\textsf

\DeclareMathOperator*{\argmin}{arg\,min}

\usepackage[margin=1in]{geometry}
\usepackage[colorlinks=true,citecolor=blue]{hyperref}

\newenvironment{keywords}
{\medskip\noindent\textbf{Keywords}\,\begin{itshape}}
{\end{itshape}}

\newcommand{\blind}{0}

\begin{document}

\title{Survival Modeling of Suicide Risk with Uncertain Diagnoses and a Cure Fraction}
\author{Wenjie Wang$^1$,
  Chongliang Luo$^2$,
  Robert H. Aseltine$^{3,4}$,
  Fei Wang$^5$,\\
  Jun Yan$^1$,
  Kun Chen$^1$\thanks{Corresponding author; kun.chen@uconn.edu}\\
  $^1$\textit{Department of Statistics, University of Connecticut}\\
  $^2$\textit{Department of Surgery, Washington University in St. Louis}\\
  $^3$\textit{Center for Population Health, UConn Health}\\
  $^4$\textit{Division of Behavioral Sciences and Community Health, UConn Health}\\
  $^5$\textit{Department of Population Health Sciences, Weill Cornell Medicine}
}

\maketitle

\begin{abstract}
  Motivated by the pressing need for suicide prevention through improving
  behavioral healthcare, we use medical claims data to study the risk of
  subsequent suicide attempts for patients who were hospitalized due to suicide
  attempts and later discharged.
  Understanding the risk behaviors of such patients at elevated suicide risk is
  an important step towards the goal of ``Zero Suicide''.
  An immediate and unconventional challenge is that the identification of
  suicide attempts from medical claims contains substantial uncertainty: almost
  20\% of ``suspected'' suicide attempts are identified from diagnosis codes
  indicating external causes of injury and poisoning with undermined intent. It is
  thus of great interest to learn which of these undetermined events are more
  likely actual suicide attempts and how to properly utilize them in survival
  analysis with severe censoring. To tackle these interrelated problems, we
  develop an integrative Cox cure model with regularization to perform survival
  regression with uncertain events and a latent cure fraction. We apply the
  proposed approach to study the risk of subsequent suicide attempts after
  suicide-related hospitalization for the adolescent and young adult population,
  using
  medical claims data from Connecticut. The identified risk factors are highly
  interpretable; more intriguingly, our method distinguishes the risk factors that
  are most helpful in assessing either susceptibility or timing of subsequent
  attempts. The predicted statuses of the uncertain attempts are further
  investigated, leading to several new insights on suicide event identification.
\end{abstract}


\section{Introduction}%
\label{sec:intr}

Suicide is a serious public health problem in the US\@. According to the
National Institute of Mental Health \citep{hedegaard2018nchsdb},
the annual suicide rate for the US population, as measured by the number of
deaths by suicide per every 100,000 population, increased 33\%
(from 10.5 to 14.0 per 100,000) from 1999 through 2017.
Losing a loved one to suicide is devastating. The effects of suicide or suicide
attempt on the survivors can be long-lasting and far-reaching.
Further, suicide is associated with high economic
costs for individuals, families, communities, and the society as a
whole~\citep{shepard2016sltb}.
Various studies show that a prior suicide attempt is a strong risk factor for
suicidal death~\citep{bostwick2015ajp}, and there
is a strong likelihood of a subsequent suicide attempt after the initial
one~\citep{suominen2004ajp,parra2017bmc}. Unfortunately, 
suicidal behavior is not always preceded by clear warnings. On the other hand,
opportunities for prevention do exist. In particular, a suicide attempter may have
been in contact with the healthcare system prior to his/her
attempt, which creates a window of opportunity for professional intervention.

Motivated by the pressing need for suicide prevention through improving
behavioral healthcare and inspired by the
{``Zero Suicide''} initiative \citep{Brodsky2018} for
transforming healthcare systems, we aim to build a data-driven approach with large-scale
medical claims data to understand and identify risk factors associated with
a subsequent suicide attempt among patients who were previously hospitalized
for suicide attempt. Being able to
identify the patients at elevated risk for subsequent attempts is an important
first step towards a better allocation of prevention efforts with limited
resources. More specifically, we examine the suicide attempt risk of youth
and young adult patients
in the State of Connecticut who had been hospitalized due
to probable or suspected suicide attempts,
using the 2012--2017 medical claims from the state's All Payers Claims
Database.

Statistically, it appears straightforward to formulate
the problem as a survival analysis, to model the time to
the subsequent suicide attempt from the
initial suicide-related hospitalization \citep{DoshiChen2020}. That is, for each patient, the event time is
observed if there was a record of a subsequent suicide attempt, and otherwise it is
considered as right censored and the censoring time is determined by the end of
the follow-up, e.g., the end of the last encounter with the healthcare system or the
end of the study period. However, 
the problem is not that straightforward,
and our attempt at adopting a conventional model
such as Cox's regression~\citep{cox1972jrssb}
becomes problematic.

The most immediate and unconventional challenge is that the
identification of suicide attempts from medical claims data carries
substantial uncertainty. The prevailing rules
for identifying suicide attempts are based on ICD-9/ICD-10
(International Classification of Diseases, 9th/10th Revision) diagnosis
codes. These include codes that directly record suicidal attempt (e.g.,
ICD-9 E950--E958/ICD-10 X71--X83: suicide and self-inflicted injury)
and some combinations of codes that are indicative of suicidal
behaviors~\citep{patrick2010pds,chen2017jah,DoshiChen2020}. While a majority of the attempts
identified by these rules can be safely considered as factual, there are also
``suspected'' attempts that hold some unignorable uncertainty.
In particular, about 19\% of the attempts in the data we examined
were identified through the ICD-9
E980--E988 codes (or their ICD-10 equivalences), meaning external
causes of injury and poisoning with undermined intent, be it accidentally or
purposely inflicted. Most existing research either
included these events without taking account of their uncertainty or simply
removed them altogether \citep{Barak2017,walsh2018jcpp,Su2020TransPsy}.
Both approaches may lead to substantial bias and/or information loss,
especially when modeling a rare event like suicide attempts.
In fact, the uncertainty in the identification of medical conditions from
diagnosis codes in claims data is quite common \citep{strom2001},
and it can be caused by various reasons including intrinsic diagnosis
uncertainty, coding errors, difference in physician practice, inaccuracy in
patient reporting, among others \citep{Bhise2018,Bell2020}.
Therefore, to better understand suicide risk and improve its predictive
modeling, there is a strong rationale for determining which events identified by
the E98 codes are actual suicide attempts.

%

Besides the uncertainty in event identification, there are several other
challenges in suicide risk modeling, including the rarity of attempts and
the large-dimensionality of candidate predictors. The rarity of suicide attempts,
even among patients with previous attempts, directly translates to a very high
censoring rate in the observed data, for which conventional survival regression
methods may lack power in making inference on the effects of risk factors and
may predict poorly.
This difficulty may partly explain why most existing research on suicide risk
modeling has moved away from survival analysis,
opting instead for a less ambiguous goal of modeling
the occurrence of attempt or death with classification
methods~\citep{Belsher2019,Kessler2020,Su2020TransPsy}.
We argue, however, that it could be beneficial to
combine classification and survival
analysis. In particular, the so-called cure
model~\citep{berkson1952jasa,Peng2000,Amico2018}
can be attractive. The approach incorporates a ``cured'' sub-population,
which is not subject to or has negligible risk
of the outcome event of interest.
In our study, it is indeed plausible that some patients are
not exposed to the risk of a subsequent suicide attempt or can be considered
``long-term survivors'' of suicidal behavior. Since our study
cohort consists of patients who were hospitalized due to either a probable or
suspected suicide attempt or
self-injury of undetermined intent,
it is possible that some of these patients never actually attempted suicide.
Another rationale for considering the cure fraction is that it can help to
evaluate whether hospitalization and/or intervention following the initial
attempt are effective in promoting long-term survival.
Therefore, it is of great interest to be able to differentiate and understand
the cured sub-population, under this
unique setting of uncertain suicide attempts.
Furthermore, there are many candidate risk factors extracted from medical claims
data such as the ICD-9/10 diagnosis codes,
making variable selection a necessity.



Figure~\ref{fig:diagram} summarizes the survival analysis setup where
uncertainty, censoring, and a latent cure fraction are simultaneously present.
A decomposition of the subjects to different categories helps to
disentangle the puzzle.
In the suicide risk study, Case~1 consists of patients for whom
the events of suicide attempt are observed with certainty,
e.g., determined by the presence of E95 codes in ICD-9.
As such, there is no uncertainty or cure
fraction in Case~1.  Case~2 consists of patients whose event times are
censored. Among those patients, some are
censored under risk (Case~2a),
while some are considered cured without exposure to the risk (Case~2b).
Lastly, Case~3 consists of patients having
uncertainty or events with undetermined cause.
In our study, Case~3 refers to patients who were
diagnosed with the E98x codes. Such events can truly be suicide
attempts (Case~3a), a misdiagnosis with risk of subsequent
attempt (Case~3b), or a misdiagnosis without subsequent risk (Case~3c).
It should be clear that the decomposition of Case~2 and of Case~3 is not
observed and have to be inferred from proper modeling of the observed data.

\begin{figure}
  \centering
  \includegraphics[width=\linewidth]{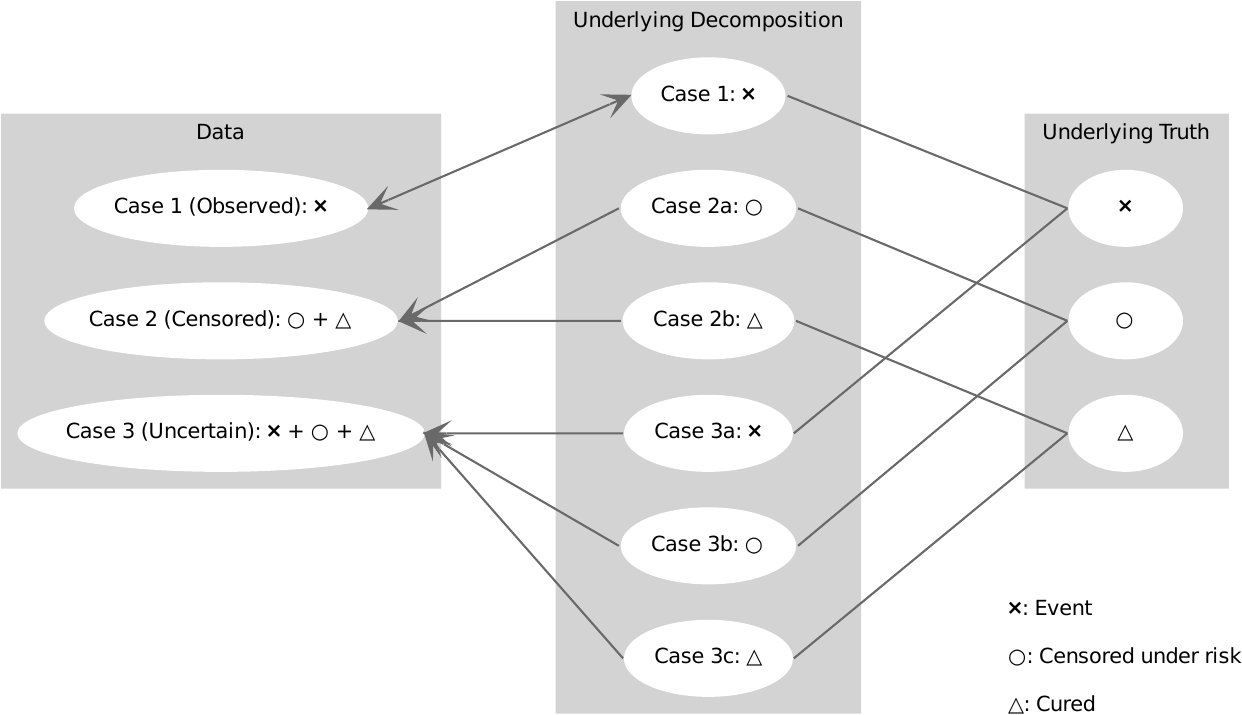}
  \caption[Decomposition of subjects.]
  {Decomposition of subjects in three different cases.}\label{fig:diagram}
\end{figure}

We propose an integrative Cox cure model with regularization to perform survival
regression with uncertain events and a latent cure
fraction. The model setup thus extends both the Cox model and the cure
model. More specifically, our approach can be regarded as a mixture model
with the regular Cox
model as its skeleton, in which the event uncertainty and the cure status are
modeled as latent variables with missing values. 
Regularized estimation
techniques 
are utilized to enable shrinkage
estimation and variable
selection. We developed a computational algorithm based on an integration of the
Majorization-Minimization algorithm, the profile likelihood, and the
coordinate descent method. 
The algorithm was shown to
be stable, efficient, and have monotonic descending property. Our new method
outperformed competing methods in simulation studies under realistic
settings similar to the motivating application of the suicide risk study.

The proposed approach was applied
to study the risk of subsequent suicide attempt after suicide-related
hospitalization for adolescent and young adult population, using medical claims
data of year 2012 to 2017 from Connecticut All-Payer Claims Database (APCD).
The identified risk factors are highly interpretable and consistent with current
understanding; intriguingly, our method is able to
distinguish the risk factors that are mostly helpful in assessing whether the patients
are under risk of subsequent attempt and the ones that can actually predict the
time of subsequent attempt. The predicted cure and event status among those
uncertain suicide attempts identified by the E98x codes are further investigated,
providing several new insights on suicide risk identification and prevention.

The rest of the paper is organized as follows. In
Section~\ref{sec:conn-suic-attempt}, we describe the claims data and the problem
setup for studying subsequent suicide attempt. In Section~\ref{sec:mode}, based
on a methodological review, we
propose the integrative Cox cure model with
uncertain events and derive its likelihood. The estimation procedure is developed in Section~\ref{sec:estimation}.  Simulation studies are presented in
Section~\ref{sec:simu}.  The suicide risk study is reported in
Section~\ref{sec:conn-suic}. Section~\ref{sec:discussion} concludes with a discussion.

\section{Data and Exploratory Analysis}\label{sec:conn-suic-attempt}

We focused on young patients of age 10--24 who were admitted to hospitals in
Connecticut from 2012--2017 due to suicide attempts, whether
determined or suspected.
Data on primary and secondary diagnoses were available from the
Connecticut APCD\@.  We excluded a small proportion of
patients with expired/dead status at discharge of their first recorded
suicide-related hospitalization, and focused on those patients
who survived and were discharged.
The event of interest was a subsequent suicide
attempt after the first hospitalization due to possible
suicidal behaviors (diagnosed by E95x or E98x codes).
The available APCD data was from October 1, 2012 to
September 30, 2017. As such, we considered a retrospective follow-up study
setup, in which the patients were followed up until a determined or suspected suicide
attempt occurred or until September 20, 2017 if no suicide attempt was observed.



A total of 7,552 patients with prior hospitalizations for suicide attempts were
included in this analysis.  Among them, 3,831 patients were female and 3,721
patients were male.  A total of 736 patients were coded as having subsequently
attempted suicide using the E95x codes, while 173 patients experienced injuries
whose cause was undetermined as to whether a suicide attempt or accident as
indicated by the E98x diagnosis codes, and the remaining 6,643 subjects did
not have a recorded suicide attempt during the follow-up.
In other words, the size of Case~1--3 is, respectively, 736, 6,643, and 173.
The censoring rate among subjects in Case~1--2 is 90.0\%, and almost
20\% of suicide attempts were identified with uncertainty.

The APCD data contained a large amount of information on the characteristics of
patients and their previous hospital admissions.  The diagnoses were mainly
recorded as ICD-9 diagnosis codes prior to fiscal year 2015, and
ICD-10 codes (the 10th Revision) were used afterwards. As the existing
rules~\citep{patrick2010pds,chen2017jah} for identifying determined/suspected
suicide
attempts were mainly based on ICD-9 codes, we translated all the ICD-10
diagnosis codes to their ICD-9 equivalence by the General Equivalence Mappings
developed from Centers by the Medicare and Medicaid Services.  Both forward
and backward mapping were used in the crosswalk as suggested by the Agency for
Healthcare Research and Quality. The translation was efficiently done by
\proglang{R} package \pkg{touch}~\citep{wang2018touch}.

After harmonizing to ICD-9 codes, we grouped the codes in the data by their
three leading characters, which resulted in 911 major diagnosis categories.
For each patient, we counted the number of appearances of the diagnosis
codes belonging to each category in his/her historical records up to the initial
admission.
We further filtered out rare diagnosis codes to
avoid separation/semi-separation problem
by restricting minimum cell counts of
the 2 by 2 contingency table of the diagnosis indicator and the event indicator
of subjects in Case~1--2 to be at least 10. While this simple filtering approach is a common practice, we remark that it is possible to apply some newly developed feature aggregation methods \citep{yan2021rare,ChenAseltineChen2022} to better utilize very rare diagnosis codes.
The remaining 246 ICD-9 categories, in
addition to demographic covariate gender and age, were considered in the
survival analysis predicting a subsequent attempt.

For subjects in Cases~1--2, the overall survival probability estimated by the
Kaplan--Meier estimator~\citep{kaplan1958jasa} with point-wise confidence
intervals (based on normality of the logarithm of survival rate estimates) is
plotted in Figure~\ref{fig:data:surv}.  A flat tail was observed at the end of
the survival curve, which suggested a long enough follow-up time and a possible
presence of a cure fraction.  We applied the method proposed
by~\citet{maller1994jasa} to formally test whether the follow-up is sufficient.
The test statistic $a_n$ was proposed to be ${(1 - N_n / n)}^n$, where $n$ is
the sample size, $N_n$ is
the number of uncensored observations in time period $(2T_n^* - T_n, T_n^*]$,
$T_n^*$ is the largest uncensored event time, and $T_n$ is the largest survival
observed time.  In our case, the test statistics $a_n$ computed based on the data
excluding Case~3 is 0.25\%, which suggested a strong evidence of a sufficient
follow-up and motivated us to consider a potential cure fraction in the
modeling.  We further performed the likelihood ratio test proposed
by~\citet{maller1995bimetrics} for testing the presence of a cure fraction among
patients in Case~1--2.  More specifically, we fitted a regular Cox model and Cox
cure model with age and gender as control variables to patients in Case~1--2,
where the logistic part of the Cox cure model only included an intercept term.
The test statistic was 3678.0, much larger than the critical value 2.71, which
provided a strong evidence of the presence of a cure fraction among patients in
Cases~1--2.

\begin{figure}
  \centering
  \includegraphics[width=\linewidth]{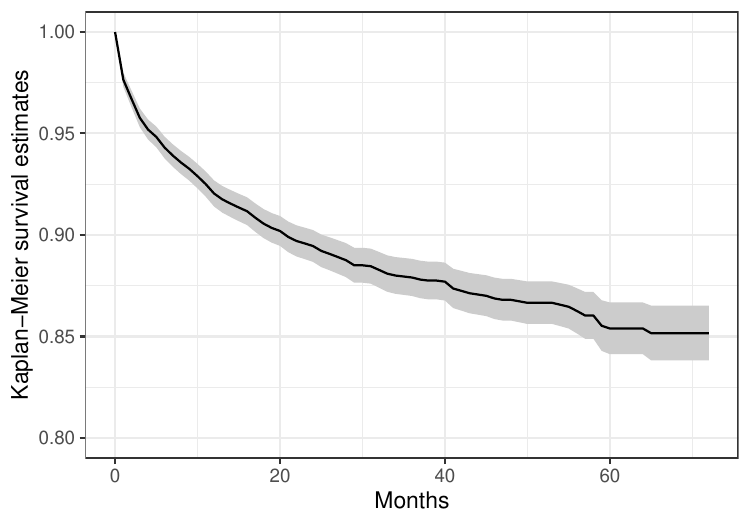}
  \caption[The estimated survival curve.]
  {The overall survival curve by the Kaplan-Meier estimator for subjects in
    Cases~1--2.}%
  \label{fig:data:surv}
\end{figure}

The unique features of the problem
demonstrated in our exploratory analysis motivate us to develop a general
statistical approach on handling uncertain and rare medical diagnosis in
time-to-event modeling.

\section{Integrative Cox Cure Model with Uncertain Events}%
\label{sec:mode}

\subsection{Overview of Existing Methods}

In the literature, a few methods have been proposed to similar problems of
mis-measured survival outcomes or uncertain events. When a binary diagnosis
outcome was measured with uncertainty, \citet{richardson2000biostatistics}
proposed an estimation procedure for the product limit estimate of survival
function with no covariate based on the expectation maximization (EM) algorithm
\citep{dempster1977jrssb}.
The method is only applicable for discrete-time contexts where the time points
of outcome testing are predetermined.
\citet{meier2003biometrics} extended the discrete proportional hazards model
\citep{kalbfleisch2002wiley} to mis-measured outcomes under a setting similar to
\citet{richardson2000biostatistics} but allowed covariate effects.
\citet{wang2020aoas} proposed a model for survival data with uncertain event
times arising from partial linkage of different datasets, i.e., for a subject
there could be multiple conflicting potential event times and censoring time,
arising from linking different datasets without unique identifiers and among
which only one of them is the truth.
None of the above methods is appropriate for the uncertainty scenario in our
problem, and none of them considered cure fraction and feature selection.

On the other hand, the cure rate model was first proposed by
\citet{berkson1952jasa}. The logit link has been widely used in modeling the
susceptible probability as proposed by \citet{farewell1982biometrics}.
\citet{kuk1992biometrika} proposed modeling the conditional survival times
through a Cox proportional hazards model~\citep{cox1972jrssb}. Its model
estimation based on the EM algorithm was later proposed by
\citet{sy2000biometrics}.
So far limited efforts have been made for survival data with masked events and
the presence of a cure fraction.
\citet{dahlberg2007biometrics} and \citet{zhang2009jmmb} proposed cure models
with masked events based on the Cox model and the accelerated failure time (AFT)
model \citep{wei1992sim}, respectively, in the framework of competing risks.
Unfortunately, neither of the above methods considered uncertain events, where
no unmasking event is observed.

Regularization techniques are often adopted in practice to conduct feature
selection with large data.
Here our review shall mainly focus on the
computational and methodological aspects of regularization that are related to
survival analysis. Coordinate descent (CD) methods for solving the lasso problem
\citep{tibshirani1996jrssb} was proposed by \citet{fu1998jcgs} and known as the
``shooting'' procedure.
The CD algorithms for generalized linear models and
Cox models were developed by \citet{friedman2010jss} and
\citet{simon2011jss}, respectively.
However, as mentioned in
\citet{friedman2010jss}, the convergence of Newton's method used for handling
non-quadratic loss is generally expected but not
guaranteed.  
For regularized Cox model,
\citet{simon2011jss} used an approximated Hessian matrix of the (partial)
log-likelihood, which also lacks convergence guarantee.
To resolve these convergence issues, \citet{yang2013sii} derived the
coordinate-majorization-descent (CMD) algorithm for the regularized Cox's model,
which combines the CD method with the principle of majorization-minimization
(MM) algorithm \citep{hunter2004as,lange2000jcgs} and always converges.
Recently, a few works have been proposed to
perform variable selection for cure models.  \citet{scolas2016sim} proposed
variable selection with adaptive lasso for interval-censored
data in a parametric cure model. 
\citet{masud2018smimr} proposed variable selection methods for mixture cure
model and promotion cure model. 
Some recent works on cure models have promoted structural
similarity \citep{fan2017smmr} and sign consistency \citep{shi2019smmr} of
covariate coefficient estimates on both model components of the cure model.


\subsection{Proposed Method}

Consider a random sample of $n$ subjects who fall into the three cases as
illustrated in Figure~\ref{fig:diagram}.
Let $I_1$, $I_2$, and $I_3$ be the indices of the subjects in Case~1, 2, and 3,
respectively.
Define $T_j=\min(V_j, C_j)$, $\Delta_j=\bm{1}(V_j \le C_j)$ with realizations
$t_j$ and $\delta_j$, where $V_j$ and $C_j$ is random variable of the event time
and the censoring time of subject $j$, $j \in \{1,\ldots,n\}$, respectively.
The survival time $T_j$ is observed for all cases.
However, the event indicator $\Delta_j$ is observed only for $j \in I_1 \cup
I_2$ but missing for $j \in I_3$.
Let $\xi_j=\bm{1}(j \in I_1\cup I_2)$ be an indicator variable taking value 1 if
$\Delta_j$ is observed and taking value 0 otherwise.
Define $Z_j = 1$ if subject $j$ is susceptible (not cured), and $Z_j = 0$
otherwise with probability $p_j = \Pr(Z_j = 1)$.
We assume that censoring time $C_j$ is independent of $Z_j$.
  Note that $Z_j = 1$ is observed for $j \in I_1$ but missing for $j \in I_2 \cup
  I_3$, which suggests that we observe $\xi_j=1$, $\Delta_j = 1$, and $Z_j=1$ for $j\in I_1$.
  Given that $Z_j = 0$, that $\Delta_j = 0$ holds with probability one.
  For $j\in I_2$, we observe $\xi_j=1$ and $\Delta_j = 0$.
  Suppose that the available data of subject $j$ are $\{(\bm{x}_j^{\top}, t_j,
  \sigma_j, \xi_j)\}$,
  where $\bm{x}_j$ is a $p$-dimensional covariate vector, the actual
  event indicator $\delta_j$ is masked by $\sigma_j= z_j \delta_j \xi_j$,
  and $z_j$ is the realization of $Z_j$.
  Given that $\xi_j=1$, we have $j \in I_1\cup I_2$ by definition
    and we observe $\sigma_j=1$ for $j\in I_1$, and $\sigma_j=0$ for $j\in I_2$.
We assume $\delta_j$ is missing at random (MAR) satisfying
\begin{align*}
  q_j := \Pr(\xi_j = 1\mid T_j = t_j, \Delta_j = \delta_j, Z_j = z_j, \bx_j) =
  \Pr(\xi_j = 1\mid T_j = t_j, \bx_j),
\end{align*}
which means the event missing indicator $\xi_j$ depends on
the fully observed $t_j$ and $\bx_j$ only.


We regard the problem as a label-missing problem and give the likelihood
function under the complete data as follows:
\begin{align*}
  & \prod_{j \in I_1} q_j p_j f(t_j\mid Z_j=1, \bx_j)\overline{G}(t_j)
  \prod_{j \in I_2} q_j g(t_j) {\left[p_j S(t_j\mid Z_j=1, \bx_j)\right]}^{z_j}
  {\left[(1 - p_j) \right]}^{1 - z_j}\\
  & \prod_{j \in I_3} (1 - q_j)
    {\left[p_j f(t_j\mid Z_j=1,\bx_j) \overline{G}(t_j)\right]}^{\delta_j z_j}
    {\left[p_j g(t_j) S(t_j\mid Z_j=1,\bx_j)\right]}^{(1 - \delta_j) z_j}
    {\left[(1 - p_j) g(t_j)\right]}^{1 - z_j}
\end{align*}
where $f(t\mid Z_j=1, \bx_j)$
and $S(t\mid Z_j=1, \bx_j)$ are the density function and survival
function of $V_j$ given the covariates and that the subject $j$ is susceptible,
respectively;
$g(t)$ and $\overline{G}(t)$ are the density function and survival function of
$C_j$, respectively.
Due to the MAR assumption, we treat the likelihood contributions of the
$\{\xi_j\}$ as nuisance components and simplify the complete-data likelihood
function as follows:
\begin{align}\label{eqn:completel}
  & \prod_{j \in I_1} p_j f(t_j\mid Z_j=1, \bx_j)\overline{G}(t_j)
  \prod_{j \in I_2} g(t_j) {\left[p_j S(t_j\mid Z_j=1, \bx_j)\right]}^{z_j}
  {\left[(1 - p_j) \right]}^{1 - z_j}\notag\\
  & \prod_{j \in I_3}
    {\left[p_j f(t\mid Z_j=1, \bx_j) \overline{G}(t_j)\right]}^{\delta_j z_j}
    {\left[p_j g(t_j) S(t_j\mid Z_j=1,\bx_j)\right]}^{(1 - \delta_j) z_j}
    {\left[(1 - p_j) g(t_j)\right]}^{1 - z_j},
\end{align}
which is useful for deriving an EM algorithm for computation.
See Supplementary Materials for the derivation details.

Given that $Z_j = 1$, we assume the event time $V_j$ of subject~$j$ follows a
Cox proportional hazards model \citep{cox1972jrssb} with the hazard function
\begin{equation}\label{eqn:cox}
  h(t\mid Z_j = 1, \bx_j) = h_0(t\mid Z_j = 1) \exp(\bx_j^{\top} \bb),
\end{equation}
where $h_0(\cdot\mid Z_j = 1)$ is an unspecified baseline function for events,
and $\bb$ is a vector of unknown coefficient of the covariate vector $\bx_j$.
Then the conditional survival function of the event time of subject~$j$ is
$S(t\mid Z_j = 1, \bx_j) = \exp\{-H_0(t\mid Z_j = 1) \exp(\bx_j^{\top} \bb) \}$,
where $H_0(t\mid Z_j = 1) = \int_0^t h_0(s\mid Z_j=1) \dif s$ and the
conditional density function is $f(t\mid Z_j = 1, \bx_j) = h(t\mid Z_j = 1,
\bx_j) S(t\mid Z_j = 1, \bx_j)$.
Given that subject $j$ is cured ($Z_j = 0$), the conditional survival function
satisfies $S(t\mid Z_j=0,\bx_j) = 1$, for $t<+\infty$.

Similarly, let $h_c(\cdot)$ denote the hazard function of censoring times and
$H_c(t) = \int_0^t h_c(s) \dif s$.
Then $g(t) = h_c(t) \overline{G}(t)$ and $\overline{G}(t) = \exp\{- H_c(t)\}$.
Here we assume the censoring time is independent of the event times and
susceptible indicators conditional on the covariates $\bx_j$.
This conditional independence assumption of the censoring time is justified for
our study because the censoring was administrative.

We utilize a logistic regression setup \citep{farewell1982biometrics} to model
$p_j$ with an intercept and covariates $\bx_j$'s. Here by allowing zero
covariate coefficients, the same set
of covariates is considered for modeling both the Cox
regression part~\eqref{eqn:cox} and the cure probability part. The logistic
model can be expressed as
\begin{align}\label{eqn:logis}
  \Pr(Z_j = 1) = p_j = \frac{\exp(\bx_j^{\top}\bg + \gamma_0)}
  {1 + \exp(\bx_j^{\top}\bg + \gamma_0)},
\end{align}
where $\gamma_0$ is the coefficient of the intercept term.

Let
$\bm{\theta} = \left\{\bg, \bb, \gamma_0, h_0(\cdot\mid Z_j=1),
  h_c(\cdot)\right\}$ be the set of unknown parameters.
Then from the complete-data likelihood in \eqref{eqn:completel}, by integrating out the missing components, the likelihood function under the observed data
$\bm{Y}^{O}=\left\{\bm{Y}_1, \bm{Y}^{O}_2, \bm{Y}^{O}_3\right\}$ is
\begin{align}\label{eqn:obs-like}
  L(\bm{\theta} \mid \bm{Y}^{O}) =
  & \prod_{j\in I_1} \left[p_j f(t_j\mid Z_j = 1, \bx_j)\overline{G}(t_j)\right]\nonumber\\
  & \prod_{j\in I_2} \left[p_j g(t_j) S(t_j\mid Z_j = 1, \bx_j) + (1 - p_j) g(t_j)\right]
    \nonumber \\
  & \prod_{j\in I_3} \large[p_j f(t_j\mid Z_j = 1, \bx_j) \overline{G}(t_j) +\nonumber\\
  &~~~~~~ p_j g(t_j) S(t_j\mid Z_j = 1, \bx_j) + (1 - p_j) g(t_j)\large].
\end{align}

With the above likelihood derivations,
we propose the following regularized likelihood estimator
\begin{align}\label{eqn:obs-like-reg}
  \hat{\bm{\theta}} = \argmin_{\bm{\theta}} -\frac{1}{n}
  \ell(\bm{\theta} \mid \bm{Y}^{O}) + P(\bm{\theta}; \lambda),
\end{align}%
where $\ell(\bm{\theta} \mid \bm{Y}^{O})$ is the log-likelihood
function under the observed data given in~\eqref{eqn:obs-like},
and $P(\bm{\theta}; \lambda)$ represents the penalty function for $\bm{\theta}$ with
the tuning parameter $\lambda$.
More specifically, we consider
the elastic net penalties~\citep{zouHastie2005jrssb}
for both $\bb$ and $\bg$, and let
$P(\bm{\theta};\lambda)=
P_{1}(\bb; \alpha_1, \lambda_1) + P_{2}(\bg; \alpha_2, \lambda_2)$,
where
\begin{align*}
  P_{1}(\bb; \alpha_1, \lambda_1)
  & = \lambda_1 \left( \alpha_1 \sum_{k=1}^{p} \omega_k \lvert \beta_k \rvert +
  \frac{1 - \alpha_1}{2} \sum_{k=1}^{p} \beta_k^2\right),\\
  P_{2}(\bg; \alpha_2, \lambda_2)
  & = \lambda_2 \left( \alpha_2 \sum_{k=1}^{p} \nu_k \lvert \gamma_k \rvert +
  \frac{1 - \alpha_2}{2} \sum_{k=1}^{p} \gamma_k^2 \right),
\end{align*}%
with tuning parameters ($\alpha_1$, $\lambda_1$) and ($\alpha_2$, $\lambda_2$)
for $\bb$ and $\bg$, respectively.
In addition, the $\omega_k$ and $\nu_k$, $k = 1, \ldots, p$, represent optional
non-negative weights~\citep{zou2009}.
For example, in the low-dimensional case where the non-regularized estimators
$(\bb^0, \bg^0)$ are reliable, one could set $\omega_k = |\beta_k^0|^{-1}$ and
$\nu_k = |\gamma_k^0|^{-1}$.
In our study, we adopted unit weights $\omega_k = 1$ and $\nu_k=1$ for simplicity.


\section{Model Estimation}%
\label{sec:estimation}

\subsection{Computational Algorithm}\label{sec:em}

We propose an estimation procedure to solve \eqref{eqn:obs-like-reg} by utilizing
the architecture of the EM algorithm, in which the M-step adopts a profile
likelihood similar to the partial likelihood~\citep{cox1975biometrika} and the
CMD algorithm~\citep{yang2013sii}.

We briefly describe the structure of the proposed algorithm.
In the E-step, we compute the conditional expectation of the complete-data log-likelihood given the
observed data and estimates from the last step. 
In the M-step, the conditional expectation of the
complete-data log-likelihood are decomposed into several
parts that involve exclusive sets of parameters. Specifically, after the hazard functions are profiled out using
``Breslow estimator'' \citep{breslow1974biometrics} and the idea of partial
likelihood of \citet{cox1975biometrika}, the problem boils down to two
parts involving $\bb$ and ($\bg$,$\gamma_0$), respectively. We adopt the
monotonic quadratic approximation \citep{bohning1988aism} and derive
CMD algorithms \citep{yang2013sii} for optimizing these two sub-problems.
We show the descending properties of the proposed updating steps and
consequently conclude that our proposed algorithm enjoys the monotonic
descending property in principle of an MM algorithm.

We summarize the resulting EM algorithm in Algorithm~\ref{alg:est}, with the
sub-routine algorithms of $\bb$, and ($\bg$,$\gamma_0$) in
Algorithms~\ref{alg:mstep-beta} and~\ref{alg:mstep-gamma},
respectively.
To simplify the notations, we define $h_j:=h(t_j\mid Z_j=1,\bx_j)$,
$S_j:=S(t\mid Z_j=1,\bx_j)$.
Most necessary quantities are defined within Algorithm~\ref{alg:est}.
We define some additional quantities below.
For simplicity, the dependency on the last step estimates is omitted.
The M-step objective functions for $\bb$ and ($\bg$,$\gamma_0$) are respectively
\begin{align}
  J(\bb)
  & = -\frac{1}{n} l(\bb) +
    P_{1}(\bb; \alpha_1, \lambda_1),\label{eqn:cox:reg}\\
  J(\bg, \gamma_0)
  & = -\frac{1}{n} l(\bg, \gamma_0) +
    P_{2}(\bg; \alpha_2, \lambda_2)\label{eqn:cure:reg},
\end{align}
where
\begin{align*}
  l(\bb)
  & = \sum_{j=1}^n \int_0^{\infty} I_j(\bb, t)
    \dif N_j(t),\\
  I_j(\bb, t)
  & = \bx_j^{\top} \bb -
    \log \left[\sum_{j=1}^n \bm{1}(t \le t_j) M_j \exp(\bx_{j}^{\top}
    \bb)\right],\\
  N_j(t)
  & = \bm{1}(t_j\le t) \left[\bm{1}(j \in I_1) +
    w_{j,1}\bm{1}(j\in I_3)\right],\\
  l(\bg, \gamma_0)
  & = \sum_{j=1}^n (\bx_j^{\top} \bg + \gamma_0) M_j
    - \log \left[1 + \exp(\bx_j^{\top} \bg + \gamma_0)\right],\\
    M_j
  & = \bm{1}(j \in I_1) + v_j\bm{1}(j\in I_2) + (w_{j,1} +
    w_{j,2}) \bm{1}(j \in I_3),\\
  v_j
  & = p_j S_j / [p_j S_j + (1 - p_j)],\\
  m_j
  & = p_j h_j S_j \overline{G}_j +
    p_j S_j g_j + (1 - p_j)g_j,\\
  w_{j,1}
  & = p_j h_j S_j\overline{G}_j/ m_j,~
  w_{j,2} = p_j S_j g_j/m_j.
\end{align*}
Let $\dot{l}_k(\bb)$ and $\dot{l}_{k'}(\bg, \gamma_0)$ denote the first
partial derivative of $l(\bb)$ and $l(\bg, \gamma_0)$ with respect to
$\beta_k$ and $\gamma_{k'}$, respectively, where $k\in\{1,\ldots,p\}$ and
$k'\in\{0,\ldots,p\}$.
Define
\begin{align*}
  D_k = \frac{1}{4n} \sum_{j=1}^n \int_0^{\infty}
  {\left[\max_{j\in \mathcal{R}(t)} (x_{j,k})
  - \min_{j\in \mathcal{R}(t)} (x_{j,k})\right]}^2
  \dif N_j(t),~
   B_k = \frac{1}{4n} \sum_{j=1}^{n} x_{j,k}^2,
\end{align*}%
where $\mathcal{R}(t) = \{j\mid Y_j(t) = 1\}$ represents the index set of
subjects in the risk-set at time $t$.
See Supplementary Materials for details.

\begin{algorithm}[tbp]
  \small
  \caption{Estimation procedure for the proposed model.
  }
  \begin{algorithmic}
    \State{\bf initialize} $\bb$, $\bg$, $\gamma_0$, $h_0(\cdot\mid Z_j = 1)$,
    and $h_c(\cdot)$;
    \Repeat\
    \For{$j=1,2,\ldots,n$}
    \Comment{The E-step}
    \begin{align*}
      p_j & \gets \frac{1}{1 + \exp(- \bx_j^{\top}\bg + \gamma_0)};\\
      S_j & \gets \exp\left[-\sum_{s \le t_j}h_0(s\mid Z_j = 1)
        \exp(\bx^{\top}_j\bb)\right];~
      \overline{G}_j \gets \exp\left[-\sum_{s \le t_j} h_c(s)\right];\\
      m_{j,1} & \gets p_j h_j S_j \overline{G}_j;~
      m_{j,2} \gets p_j h_c S_j \overline{G}_j;~
      m_{j,3} \gets (1 - p_j) h_c \overline{G}_j;~
      v_j \gets \frac{m_{j,2}}{m_{j,2} + m_{j,3}};
    \end{align*}
    \For{$k=1,2,3$}
    \begin{align*}
      w_{j,k} \gets \frac{m_{j,k}}{\sum_{k=1}^3 m_{j,k}};
    \end{align*}
    \EndFor\
    \begin{align*}
      M_j \gets \bm{1}(j \in I_1) + v_j\bm{1}(j\in I_2) +
      (w_{j,1} + w_{j,2}) \bm{1}(j \in I_3);
    \end{align*}
    \EndFor\
    \For{each $t \in \mathcal{T}_2 \cup \mathcal{T}_3$}
    \Comment{The M-step of $h_c(\cdot)$}
    \begin{align*}
      C(t) & \gets \sum_{j=1}^n \bm{1}(t_j\le t)\bm{1}(j\in I_2) +
             (w_{j,2} + w_{j,3})\bm{1}(t_j\le t)\bm{1}(j\in I_3);\\
      h_c(t) & \gets \frac{\dif C(t)}{\sum_{j=1}^{n} Y_{j}(t)};
    \end{align*}
    \EndFor\
    \For{each $t \in \mathcal{T}_1 \cup \mathcal{T}_3$}
    \Comment{The M-step of $h_0(\cdot\mid Z_j=1)$}
    \begin{align*}
      N(t) & \gets \sum_{j=1}^n \bm{1}(t_j\le t) \bm{1}(j \in I_1) +
             w_{j,1}\bm{1}(t_j\le t)\bm{1}(j\in I_3);\\
      h_0(t\mid Z_j=1) & \gets \frac{\dif N(t)}
        {\sum_{j=1}^n \bm{1}(t \le t_j) M_j \exp(\bm{x}_{j}^{\top}\bb)};
    \end{align*}
    \EndFor\
    \State{\textbf{go to} Algorithm~\ref{alg:mstep-beta}}
    \Comment{The M-step of $\bb$}
    \State{\textbf{go to} Algorithm~\ref{alg:mstep-gamma}}
    \Comment{The M-step of $\bg$ and $\gamma_0$}
    \Until{Convergence of $\bb$, $\bg$, and $\gamma_0$}
  \end{algorithmic}\label{alg:est}
\end{algorithm}


%

\begin{algorithm}[tbp]
  \small
  \caption{The M-step of $\bb$.}
  \begin{algorithmic}
    \State{\bf input}: $\bb$ from last EM iteration
    \Repeat\
    \For{$k=1,\ldots,p$}
    \Comment{update $\beta_k$}
    \begin{align*}
      \beta_k \gets \frac{s(D_k \beta_k +
      \frac{1}{n} \dot{l}_k(\bb),
      \lambda_1 \alpha_1 \omega_k)}{D_k + \lambda_1(1 - \alpha_1)},
    \end{align*}
    \hspace{0.5cm} where $s(a, b) = {(\lvert a \rvert - b)}_{+} \mathrm{sign}(a)$ is the
soft-thresholding operator.
    \EndFor\
    \Until{Convergence of $\bb$}
  \end{algorithmic}\label{alg:mstep-beta}
\end{algorithm}

\begin{algorithm}[tbp]
  \small
  \caption{The M-step of $\bg$ and $\gamma_0$.}
  \begin{algorithmic}
    \State{\bf input}: $\bg$ and $\gamma_0$ from last EM iteration;
    \Repeat\
    \State{\textbf{do}}
    \Comment{update $\gamma_0$}
    \begin{align*}
      \gamma_0 \gets \gamma_0 + \frac{4}{n} \dot{l}_0(\bg, \gamma_0);
    \end{align*}
    \For{$k=1,\ldots,p$}
    \Comment{update $\gamma_k$}
    \begin{align*}
      \gamma_k \gets \frac{s(B_k {\gamma}_k +
      \frac{1}{n} \dot{l}_k(\bg, \gamma_0),
      \lambda_2\alpha_2 \nu_k)}{B_k + \lambda_2(1 - \alpha_2)},
    \end{align*}
    \hspace{0.5cm} where $s(a, b) = {(\lvert a \rvert - b)}_{+} \mathrm{sign}(a)$ is the
    soft-thresholding operator.
    \EndFor\
    \Until{Convergence of $\bg$ and $\gamma_0$}
  \end{algorithmic}\label{alg:mstep-gamma}
\end{algorithm}

\subsection{Initialization, Tail Completion, and Tuning}

We propose a simple yet pragmatic initialization procedure for the proposed EM
algorithm:
\begin{itemize}\setlength\itemsep{0em}
\item[(i)]
  Setting the event indicators of subjects in Case~3 to be 0.5,
  fit a regular logistic model on event indicators and use the estimated
  coefficients to initialize $\hat{\bg}$ and $\hat{\gamma_0}$;
\item[(ii)]
  Fit a regular Cox model on all the certain events (Case~1) and use the
  estimated coefficients to initialize $\hat{\bb}$;
  For subject $j$, $j\in\{1,\ldots,n\}$, initialize $\hat{S}_{j}$ with the
  fitted survival function evaluated at $t_{j}$; initialize $\hat{h}_{j}$
  with a nearest left neighbor interpolation of the fitted hazard function
  (if no left neighbor, use nearest right neighbor).
\item[(iii)]
  Switching event and censoring for all the certain records (Case~1--2),
  estimate the hazard function for censoring
  by the Nelson-Aalen estimator (without covariates) and obtain the
  corresponding survival function estimate;
  initialize $\hat{\overline{G}}_{j}$ with the fitted survival function
  evaluated at $t_{j}$; initialize $\hat{h}_c(t_{j})$ with a nearest left
  neighbor interpolation of the fitted hazard function
  (if no left neighbor, use nearest right neighbor).
\end{itemize}

The initialization procedure was applied in the simulation studies presented in
Section~\ref{sec:simu} and the results were satisfactory in all scenarios.
In addition, we applied the zero tail completion \citep{sy2000biometrics} in the
E-steps when updating the conditional survival function of event times to avoid
identifiability issue of cure models~\citep{liTaylor2001spl,hanin2014jma}, which
imposes the constraint that $S_0(t^*\mid Z=1)=0$ for any $t^*$ greater than the
largest event time onto the baseline survival function $S_0(t)$.


We select the tuning parameters, $(\lambda_1,\alpha_1,\lambda_2,\alpha_2)$,
through a grid search.
To be more specific, we generate an
equally-spaced decreasing sequence in logarithm scale of the tuning parameters,
$\lambda_1$ and $\lambda_2$, respectively, from the smallest values that produce
all zero coefficients for the corresponding model component.
Similarly, we can select $\alpha_1$ and $\alpha_2$, respectively, from an
equally-spaced sequence between 0 and 1; in practice it is also common to simply
set them as some constant close to 1.
As we have multiple tuning parameters, we propose to use the Bayesian
information criteria (BIC) for computational efficiency.
In particular, we adopted the BIC for censored survival models
\citep{volinsky2000biometrics}, where the penalty term is proportional to the
effective sample size, i.e., number of uncensored observations, instead of
number of observations.
Alternatively, one may select the final model by K-fold cross-validation,
where the observation with the largest event time must always stay in the
training set to avoid numerical issue due to zero-tail completion.
However, the tuning procedure using cross-validation can be much more
computationally intensive than using the BIC.

\if1\blind
We have made our implementation of the proposed methods available in a
user-friendly R package (name omitted for double-anonymized peer review).
\fi

\if0\blind
We have made our implementation of the proposed methods available in a
user-friendly R package \pkg{intsurv}, which can be accessed at
\href{https://cran.r-project.org/package=intsurv}{https://cran.r-project.org/package=intsurv}.
\fi

\section{Simulation Study}%
\label{sec:simu}

\subsection{Setup and Data Generation}%
\label{sec:simu:setup}

We designed simulation settings to mimic our suicide risk application. The
main concern was the  estimation/prediction performance of the proposed
model and its ability to distinguish the true events among the uncertain events.

The simulated dataset consisted of subjects in three different cases as
illustrated in Figure~\ref{fig:diagram}. 
Accordingly, we first generated the survival data with a cure fraction,
which served as the ground truth.
Then we randomly assigned each subject to different cases from
its ground truth.  To be more specific, for each subject who actually has
observed event, we assigned it randomly to Case~1 with probability $p_1$, or
Case~3a with probability $1 - p_1$.  For each susceptible subject whose event
time was censored, we randomly assign it to Case~2a with probability $p_2$, or
Case~3b with probability $1 - p_2$.  For each cured subject, we randomly assign
it to Case~2b with probability $p_3$, or Case~3c with probability $1 - p_3$.
A diagram of the data generation is provided in Supplementary Materials.

In the following simulation settings, we fixed the size of Case~1 and Case~2 to
be approximately 100 and 1,900, respectively.
The censoring rate among these certain records (Case~1--2) was thus about 95\%,
which simulated the severe censoring in the real suicide risk data.
The size of Case~3 was set to be 100 or 200, i.e., the same or twice of the size
of Case~1, which resulted in the sample size being 2,100 or 2,200,
respectively.
Within Case~3, we set the proportion of Case~3a to be 30\% or 70\%,
respectively.
The value of $p_1$ could then be computed as the size ratio of Case~3a to
Case~1.
For simplicity, we let $p_2 = p_3$.  Thus, the size ratio of Case~2 to Case~3b
\& Case~3c could be computed as $p_2 / (1 - p_2)$, which determined the value of
$p_2$ (and $p_3$).

We generated susceptible indicators from logistic model with an intercept
term.  We set the coefficient of the intercept term to $-0.847$ or $0.847$ so
that the baseline cure rate (when all covariates are zeros) was 70\% or 30\%,
respectively.  For susceptible subjects, the event times were generated from
Weibull-Cox model with baseline hazard function
$h_0(t; \bx) = 0.2t\exp(\bx^{\top}\bb)$;
For cured subjects, the event times were set to be infinity.
The censoring times were generated independently with
the event times from exponential distribution truncated at 10.  The rate
parameter of the truncated exponential distribution is tuned so that the desired
case's decomposition was attained.

The different specification on the baseline cure rate, the size of Case~3, and
the proportion of Case~3a resulted in totally eight different simulation
scenarios.  More specifically, the size of Case~3 was set to be the same with
the size of Case~1 in the scenario (1) to (4) and twice in the scenario (5) to
(8); The proportion of Case~3a was set to be 30\% in the scenario (1), (2), (5),
and (6), and 70\% in the remaining scenarios; The baseline cure rate was set to
be 30\% in scenarios (1), (3), (5), and (7), and 70\% in the remaining
scenarios. We applied these eight scenarios to a low-dimensional setting ($p=8$) without
regularization and a large-dimensional setting ($p=100$) with regularization in Section~\ref{sec:simu:low} and
Section~\ref{sec:simu:high}, respectively. The number of replicates for each
scenario was 1,000.


\subsection{Competing Methods}%
\label{sec:comp-meth-eval}

The proposed method is denoted by I.Cure.
We considered three competing approaches based on the Cox cure rate model.
The first approach (denoted by Cure1) fits the Cox cure model to those certain
records (Case~1--2) only, which simply excludes subjects with uncertain events.
The second approach (denoted by Cure2) fits the Cox cure model for all subjects,
where subjects having uncertain events in Case~3 are all treated as censored.
In contrast, the third method (denoted by Cure3) fits the Cox cure model by
taking the uncertain events all as actual events ignoring their uncertainty.
We did not include the Cox models as they performed worse than their
corresponding Cox cure rate models in the presence of a cure fraction.

We also included an oracle procedure (denoted
by O.Cure) based on Cox cure rate model where the true event indicators in Case~3
were all given.
The oracle procedure is infeasible in practice, but it provides a reference on
the best achievable performances in the comparison.

\subsection{Evaluation Metrics}

For both model components, we measured the estimation performance by the
$\ell_2$-norm of $(\hat{\bb} - \bb^0)$, i.e.,
$\|\hat{\bb}-\bb^0\| = {[{(\hat{\bb}-\bb^0)}^{\top}(\hat{\bb}-\bb^0)]}^{1/2}$,
where $\hat{\bb}$ and $\bb^0$ represent the estimated covariate coefficients
and the underlying true covariate coefficients, respectively.

For subject $j$ in Case~3, we took estimated posterior probability $w_{j,1}$
from Algorithm~\ref{alg:est}, i.e., the estimated probability of subject $j$
having actual event at time $t_j$, for the identification of uncertain events.
We also computed the oracle posterior probability based on Bayes rule for such
identification using the underlying true models and simulated datasets,
providing references of the best achievable performances.  Given the simulated
true event indicators in Case~3, we were able to evaluate the identification
correctness by the area under curve (AUC) of the receiver operating
characteristic curve (ROC).

For large-dimensional models, we were interested in the variable selection
performance of regularized estimation.
The performance was measured by the true positive rate (TPR) among those
non-zero coefficients and the false positive rate (FPR) among those zero
coefficients, which were defined as follows:
\begin{align*}
   \mathrm{TPR} = \frac{\sum_{k=1}^p
  \bm{1}(\hat{\beta}_k  \neq 0) \bm{1}(\beta_k^0 \neq 0)}
  {\sum_{k=1}^p \bm{1}(\beta_k^0 \neq 0)},~ \mathrm{FPR} = \frac{\sum_{k=1}^p
  \bm{1}(\hat{\beta}_k  \neq 0) \bm{1}(\beta_k^0 = 0)}
  {\sum_{k=1}^p \bm{1}(\beta_k^0 = 0)},
\end{align*}
where $\beta_k^0$ is the $k$-th underlying truth covariate coefficient,
and $\hat{\beta}_k$ is the $k$-th estimated covariate coefficient from the
regularized estimation procedure. 

\subsection{Setting 1}%
\label{sec:simu:low}

In Setting 1, we considered $p=8$ covariates, $x_k$,
$k\in\{1,\ldots,8\}$, that were randomly generated from multivariate normal
distribution with marginal means zero and variances one.  The correlation
between $x_k$ and $x_l$, $k\neq l$, was set to be $\rho^{\lvert k - l \rvert}$,
where $\rho = 0.2$.  We considered two overlapped but different sets of
covariates in the logistic model and the Weibull-Cox model.  The covariates
$x_4$--$x_8$ were used in the logistic model for generating the susceptible
indicators for each subject.  For susceptible subjects, the covariates
$x_1$--$x_5$ were used in the Weibull Cox model for generating event times.
The underlying true non-zero covariate coefficients were fixed to be
$\bb^0 = \bg^0 = {(0.8, 0.5, 1, 0.3, 0.6)}^{\top}$ in both parts.  We assumed
that the underlying true sets of covariates were both known in the
low-dimensional settings.  In other words, among eight covariates, only
$x_1$--$x_5$ were considered in the Cox model~\eqref{eqn:cox} and only
$x_4$--$x_8$ were considered in the logistical model~\eqref{eqn:logis}, and no
regularization was applied in all methods.

The estimation results are presented in Table~\ref{tab:est:mod}, which show that
I.Cure gave better estimation performance in both model parts than those
competing methods and close performance compared with the oracle method in
almost all 8 scenarios.
In addition, the AUC's on identifying the true events from uncertain records
were about 75\% using I.Cure, about 13\% lower than those from the oracle method
across all scenarios; the detailed results are presented in
Supplementary Materials.



\begin{table}[tbp]
  \small\sf\centering
  \caption{Simulation results of Setting 1: comparison on parameter estimation performance for the incidence part
    and the latency part, respectively, through mean of $\|\hat{\bg}-\bg^0\|$
    and $\|\hat{\bb}-\bb^0\|$ (with the standard deviation given in
    parenthesis).}\label{tab:est:mod}
  \setlength{\tabcolsep}{2.5pt}
  \begin{tabular}{cccccccccccc}
    \toprule
             & \multicolumn{5}{c}{$\|\hat{\bg}-\bg^0\|$}
             &        & \multicolumn{5}{c}{$\|\hat{\bb}-\bb^0\|$}                                         \\
    \cmidrule(lr){2-6}\cmidrule(lr){8-12}
    \# & O.Cure & I.Cure & Cure1  & Cure2  & Cure3  &
             & O.Cure & I.Cure & Cure1  & Cure2  & Cure3                                                  \\
    \midrule
    1        & 1.05   & 1.17   & 1.28   & 1.37   & 1.20   &  & 0.26   & 0.29   & 0.31   & 0.33   & 0.70   \\
             & (0.42) & (0.45) & (0.43) & (0.44) & (0.54) &  & (0.10) & (0.11) & (0.12) & (0.12) & (0.10) \\
    [1ex]
    2        & 0.57   & 0.69   & 0.74   & 0.78   & 1.17   &  & 0.28   & 0.30   & 0.33   & 0.34   & 0.71   \\
             & (0.21) & (0.28) & (0.26) & (0.26) & (0.37) &  & (0.09) & (0.10) & (0.12) & (0.12) & (0.11) \\
    [1ex]
    3        & 0.86   & 1.11   & 1.39   & 1.54   & 0.87   &  & 0.23   & 0.30   & 0.33   & 0.35   & 0.38   \\
             & (0.33) & (0.40) & (0.41) & (0.41) & (0.38) &  & (0.08) & (0.11) & (0.12) & (0.13) & (0.10) \\
    [1ex]
    4        & 0.47   & 0.70   & 0.88   & 0.95   & 0.63   &  & 0.24   & 0.30   & 0.34   & 0.35   & 0.39   \\
             & (0.17) & (0.27) & (0.26) & (0.26) & (0.21) &  & (0.08) & (0.10) & (0.12) & (0.12) & (0.11) \\
    [1ex]
    5        & 0.93   & 1.15   & 1.35   & 1.52   & 1.34   &  & 0.24   & 0.29   & 0.33   & 0.35   & 0.87   \\
             & (0.35) & (0.45) & (0.45) & (0.44) & (0.65) &  & (0.08) & (0.10) & (0.11) & (0.12) & (0.08) \\
    [1ex]
    6        & 0.48   & 0.68   & 0.79   & 0.89   & 1.62   &  & 0.24   & 0.28   & 0.32   & 0.34   & 0.87   \\
             & (0.17) & (0.27) & (0.25) & (0.25) & (0.43) &  & (0.08) & (0.10) & (0.12) & (0.12) & (0.08) \\
    [1ex]
    7        & 0.65   & 1.03   & 1.54   & 1.81   & 0.75   &  & 0.19   & 0.33   & 0.36   & 0.38   & 0.44   \\
             & (0.23) & (0.35) & (0.38) & (0.38) & (0.27) &  & (0.07) & (0.12) & (0.12) & (0.13) & (0.08) \\
    [1ex]
    8        & 0.35   & 0.74   & 1.09   & 1.22   & 0.78   &  & 0.19   & 0.30   & 0.35   & 0.37   & 0.47   \\
             & (0.12) & (0.27) & (0.23) & (0.22) & (0.20) &  & (0.06) & (0.11) & (0.12) & (0.13) & (0.11) \\
    \bottomrule
  \end{tabular}
  \end{table}

We estimated the standard error estimates based on inter-quartile
range and normal approximation.  To check the performance of the proposed method
in making inferences about the unknown covariate coefficients, we used bootstrap
with 200 bootstrap samples.  The mean of standard error (SE) estimates and the
empirical SEs for the coefficient of all covariates are presented in
Supplementary Materials.
The bootstrap SE estimates were close to the empirical SEs of the coefficient
estimates in most of the settings.

\subsection{Setting 2}%
\label{sec:simu:high}


\begin{table}[btp]
  \small\sf\centering
  \caption{Simulation results of Setting 2: comparison on variable selection performance for the logistic model
    part and Cox model part through mean of true positive rate
    (TPR) and false positive rate (FPR) in percentage (with the standard
    deviation given in parenthesis).}\label{tab:simu:var}
  \setlength{\tabcolsep}{2.5pt}
  \begin{tabular}{ccccccccccccc}
    \toprule
             &         & \multicolumn{5}{c}{Logistic}
             &         & \multicolumn{5}{c}{Cox}                                                                    \\
    \cmidrule(lr){3-7}\cmidrule(lr){9-13}
    \# & Measure & O.Cure & I.Cure & Cure1  & Cure2  & Cure3  &  & O.Cure & I.Cure & Cure1  & Cure2  & Cure3  \\
    \midrule
    1        & TPR     & 22.0   & 29.3   & 24.3   & 22.0   & 7.4    &  & 95.4   & 93.5   & 89.5   & 89.6   & 92.1   \\
             &         & (29.7) & (31.1) & (25.3) & (12.5) & (33.3) &  & (13.1) & (21.4) & (7.9)  & (13.2) & (19.9) \\
             & FPR     & 1.14   & 1.93   & 1.52   & 1.40   & 0.41   &  & 3.29   & 4.14   & 2.98   & 2.96   & 2.41   \\
             &         & (1.96) & (2.70) & (3.16) & (0.66) & (2.44) &  & (2.68) & (3.07) & (3.68) & (2.03) & (2.80) \\
    [1ex]
    2        & TPR     & 70.5   & 66.2   & 66.5   & 65.7   & 29.5   &  & 94.7   & 92.9   & 89.6   & 87.6   & 89.5   \\
             &         & (29.4) & (35.3) & (27.3) & (29.1) & (31.7) &  & (16.7) & (19.0) & (17.1) & (7.9)  & (24.8) \\
             & FPR     & 2.42   & 3.06   & 2.54   & 2.61   & 0.73   &  & 3.24   & 4.50   & 3.20   & 3.18   & 3.05   \\
             &         & (2.44) & (1.39) & (2.68) & (2.51) & (2.47) &  & (3.46) & (2.51) & (2.98) & (2.82) & (2.69) \\
    [1ex]
    3        & TPR     & 29.1   & 42.3   & 29.7   & 26.3   & 17.0   &  & 97.5   & 96.6   & 90.0   & 87.0   & 97.5   \\
             &         & (29.0) & (33.7) & (30.6) & (28.0) & (32.1) &  & (21.7) & (7.6)  & (19.8) & (17.3) & (4.7)  \\
             & FPR     & 1.17   & 2.35   & 1.67   & 1.60   & 0.63   &  & 3.40   & 4.29   & 3.16   & 2.89   & 3.16   \\
             &         & (2.29) & (1.90) & (2.77) & (3.55) & (1.08) &  & (2.77) & (2.58) & (2.90) & (4.54) & (2.26) \\
    [1ex]
    4        & TPR     & 80.6   & 81.8   & 74.1   & 68.3   & 61.2   &  & 97.3   & 97.2   & 90.5   & 87.8   & 96.4   \\
             &         & (28.5) & (33.1) & (38.8) & (28.8) & (21.3) &  & (22.3) & (10.6) & (9.0)  & (17.6) & (3.6)  \\
             & FPR     & 2.42   & 3.64   & 2.82   & 2.55   & 1.42   &  & 3.10   & 4.40   & 3.09   & 2.99   & 3.29   \\
             &         & (2.18) & (2.55) & (2.03) & (2.64) & (2.24) &  & (2.66) & (3.46) & (2.79) & (2.75) & (2.25) \\
    [1ex]
    5        & TPR     & 27.5   & 36.1   & 26.6   & 24.3   & 3.4    &  & 97.4   & 95.2   & 91.8   & 88.7   & 95.0   \\
             &         & (17.9) & (30.6) & (33.3) & (29.9) & (14.6) &  & (19.4) & (20.3) & (9.1)  & (21.7) & (5.1)  \\
             & FPR     & 1.12   & 2.39   & 1.37   & 1.42   & 0.15   &  & 3.34   & 5.30   & 3.23   & 3.07   & 2.26   \\
             &         & (1.17) & (2.29) & (1.83) & (2.67) & (3.41) &  & (2.18) & (2.86) & (2.66) & (2.68) & (4.61) \\
    [1ex]
    6        & TPR     & 79.1   & 73.9   & 72.4   & 69.3   & 15.0   &  & 97.3   & 94.2   & 90.3   & 87.6   & 93.9   \\
             &         & (32.7) & (29.5) & (30.5) & (27.7) & (20.8) &  & (11.8) & (24.8) & (15.1) & (12.9) & (13.2) \\
             & FPR     & 2.43   & 3.82   & 2.67   & 2.60   & 0.26   &  & 3.10   & 5.62   & 2.92   & 2.77   & 3.45   \\
             &         & (2.57) & (2.39) & (2.62) & (0.85) & (2.59) &  & (2.84) & (2.64) & (4.03) & (2.37) & (2.61) \\
    [1ex]
    7        & TPR     & 43.9   & 61.3   & 42.4   & 31.0   & 21.7   &  & 99.2   & 98.8   & 89.4   & 85.7   & 98.9   \\
             &         & (32.7) & (27.9) & (29.8) & (38.2) & (22.8) &  & (17.7) & (7.8)  & (19.1) & (3.8)  & (18.2) \\
             & FPR     & 1.15   & 3.22   & 1.94   & 1.72   & 0.45   &  & 3.47   & 5.38   & 3.05   & 2.85   & 3.29   \\
             &         & (2.97) & (1.33) & (2.07) & (1.62) & (2.72) &  & (3.79) & (2.24) & (2.83) & (2.65) & (2.47) \\
    [1ex]
    8        & TPR     & 90.9   & 91.7   & 82.1   & 76.1   & 62.6   &  & 99.4   & 98.7   & 92.3   & 87.9   & 98.6   \\
             &         & (32.2) & (29.4) & (28.6) & (31.4) & (41.8) &  & (19.5) & (7.0)  & (22.5) & (5.9)  & (5.2)  \\
             & FPR     & 2.19   & 4.23   & 2.80   & 2.67   & 1.32   &  & 2.49   & 5.49   & 2.66   & 2.41   & 3.33   \\
             &         & (2.65) & (2.44) & (2.15) & (2.81) & (2.08) &  & (3.06) & (2.82) & (2.73) & (4.04) & (2.64) \\
    \bottomrule
  \end{tabular}
\end{table}

In Setting 2, we increased the number of covariates to 100.  The
covariates were randomly generated from multivariate normal distribution with
marginal means zero and variances one.  The correlation between $x_k$ and $x_l$,
$k\neq l$, $k,l\in\{1,\ldots,100\}$, was set to be $\rho^{\lvert k - l \rvert}$,
where $\rho = 0.2$. For $x_1$--$x_8$, we considered a same set of covariate
coefficients in the low-dimensional settings.  We set the coefficients of the
remaining covariates to be all zero. We conducted the proposed regularized
estimation approach with the estimation criteria given
in~\eqref{eqn:obs-like-reg}. 
The number
of the grid points was set to be 10 for both $\lambda_1$ and $\lambda_2$, which
resulted in a 10 by 10 two-dimensional grid. The minimum tuning parameter was
set to be 0.1 of the smallest value that produces all zero
coefficients. For simplicity, we set $\alpha_1 = \alpha_2 = 1$.

The variable selection results are summarized in Table~\ref{tab:simu:var}.  The
TPR's from I.Cure were evidently greater than those from the competing methods
under almost all the settings for both model components.
On the other hand, the FPR's from
I.Cure were only slightly larger than those from the competing methods.
Overall, I.Cure gave considerably high TPR's compared with the oracle method at
a cost of slightly high FPR's.



In addition, we investigated the prediction performance of different
methods using an out-of-sample procedure. The results show that I.Cure led to
consistently better out-of-sample AUC for the
prediction of susceptibility and better out-of-sample C-index for the prediction
of survival time than the
competing methods in all the scenarios.
We also considered using adaptive weights and obtained similar results compared
to using unit weights.
The detailed results are provided 
in Supplementary Materials.


\section{Understanding Subsequent Suicide Attempts with
  Uncertain Diagnosis Records}%
\label{sec:conn-suic}


We applied the proposed method and three competing methods to the Connecticut
suicide attempt data. The elastic-net penalization was used  because some
predictors could be highly correlated. Our main interests were to (1) identify
relevant diagnosis categories that were predictive of subsequent
suicide attempts after the initial hospitalization due to suicidal behaviors,
(2) verify that the proposed integrative approach has the potential of improving
predictive power over methods that do not address uncertainty, and (3) investigate
the suspected attempts among subjects in Case~3, to understand which injuries of
undetermined intent were more likely self-inflicted.\\





\noindent \textbf{Risk factor selection}

The odds ratio and hazard ratio estimates from the proposed method and
those competing ones are summarized in Table~\ref{tab:suci:vari}, in which we
also report the prevalence of each selected ICD-9 category in each of the three
cases of patients. Several
ICD-9 categories related to mental health and drug abuse were selected by the
different approaches. As expected, the prevalence of each selected ICD-9
category is the highest among Case 1 patients and the lowest among Case 2
patients. The number of variables
selected in the incidence part for modeling susceptible status ranged from 6 to
11, while only 1 to 2 variables were selected for modeling the conditional
time-to-event distribution. This suggests that understanding the timing of
suicide is generally a much harder problem than understanding its
occurrence.

\begin{table}[tbp]
  \small\sf\centering
  \caption{Suicide risk study: prevalence of selected ICD-9
    categories in three cases of patients and the exponentiated coefficient
    estimates (hazard or odd ratios) from
    different approaches.}\label{tab:suci:vari}
  \setlength{\tabcolsep}{2pt}
  \setlength{\tymin}{2pt}
  \begin{tabulary}{\textwidth}{crrrccccccccL}
    \toprule & \multicolumn{3}{c}{Prevalence \%}
             & \multicolumn{2}{c}{I.Cure}
             & \multicolumn{2}{c}{Cure1}
             & \multicolumn{2}{c}{Cure2}
             & \multicolumn{2}{c}{Cure3}                                         \\
    \cmidrule(lr){2-4}
    \cmidrule(lr){5-6} \cmidrule(lr){7-8}
    \cmidrule(lr){9-10} \cmidrule(lr){11-12}
    ICD-9    & $I_1$           & $I_2$ & $I_3$
             & $e^{\hat{\bg}}$ & $e^{\hat{\bb}}$
             & $e^{\hat{\bg}}$ & $e^{\hat{\bb}}$
             & $e^{\hat{\bg}}$ & $e^{\hat{\bb}}$
             & $e^{\hat{\bg}}$ & $e^{\hat{\bb}}$
             & Description                                                       \\
    \midrule
    296      & 55.3            & 26.0  & 29.4 & 1.78 &      & 1.80 &      & 1.77 &      & 1.77 &
             & Episodic mood disorders                                           \\
    298      & 19.2            & 7.6   & 8.9  & 1.08 &      &      &      & 1.06 &      & 1.01 &
             & Other Nonorganic Psychoses                                        \\
    300      & 47.6            & 25.8  & 28.3 & 1.10 &      & 1.04 &      & 1.10 &      & 1.02 &
             & Anxiety, Dissociative and Somatoform Disorders                    \\
    301      & 13.9            & 4.6   & 5.6  & 1.32 &      & 1.22 &      & 1.32 &      & 1.21 &
             & Personality disorders                                             \\
    304      & 12.4            & 4.0   & 5.3  & 1.31 & 1.01 & 1.11 &      & 1.17 &      & 1.60 &
             & Drug Dependence                                                   \\
    312      & 8.7             & 3.0   & 3.7  & 1.10 &      &      &      & 1.10 &      & 1.06 &
             & Disturbance of Conduct, Not Elsewhere Classified                  \\
    313      & 6.8             & 2.2   & 2.7  & 1.23 &      & 1.01 &      & 1.19 &      & 1.15 &
             & Disturbance of Emotions Specific to Childhood and Adolescence     \\
    319      & 1.6             & 0.3   & 0.4  & 1.02 &      &      &      & 1.21 &      &      &
             & Unspecified intellectual disabilities                             \\
    507      & 6.2             & 2.6   & 3.0  &      &      &      &      &      &      & 1.05 &
             & Pneumonitis Due to Solids and Liquids                             \\
    564      & 31.8            & 12.6  & 14.8 & 1.02 &      &      &      & 1.04 &      & 1.02 &
             & Functional Digestive Disorders, Not Elsewhere Classified          \\
    977      & 1.8             & 0.4   & 0.6  &      &      &      &      &      &      & 1.05 &
             & Poisoning by Other and Unspecified Drugs and Medicinal Substances \\
    V62      & 24.7            & 14.5  & 16.2 & 1.31 & 1.01 & 1.32 & 1.00 & 1.33 & 1.02 & 1.25 & 1.00
             & Other Psychosocial Circumstances                                  \\
    \bottomrule
  \end{tabulary}
\end{table}

There were 6 ICD-9 predictors selected by all methods in the incidence part.  We
found that these results were well supported by existing literature.  In
particular, all models suggested that patients with personality disorders (ICD-9
301) had a higher risk of being susceptible and making further suicide
attempts. Among these patients, about 35.2\% had
borderline personality disorder (ICD-9 301.83) and 15.9\% had chronic
depressive personality disorder (ICD-9 301.12) in our data.
Similar finding had been reported by e.g., \citet{harris1997bjp},
\citet{lieb2004lancet} and \citet{mcgirr2007jcp}, among others.
Another interesting selected ICD-9 category was V62 indicating psychosocial
circumstances. The majority (66.3\%) of the patients with V62 had
suicide ideations (ICD-9 V62.84),
which was one of the conditions for identifying determined suicide attempts in
addition to the E95 codes. Among patients with suicide
ideation diagnosis, about 88.2\% were diagnosed at the initial hospitalization,
which strongly suggested that initial hospitalization could provide important
information relevant to the prevention of further suicide attempts.
Other commonly selected ICD-9 categories were related
to mood disorders (ICD-9 296), anxiety (ICD-9 300), drug dependence (ICD-9 304),
and disturbance of emotions (313), which are well known to be important suicide
risk factors. Besides the aforementioned 6 categories, our proposed
method selected 4 more, including nonorganic psychoses (ICD-9 298), disturbance
of conduct (ICD-9 312), unspecified intellectual disabilities (ICD-9 319), and
functional digestive disorders (ICD-9 564). The associations between these
conditions and suicide risk are well supported by existing studies
\citep{Falcone2010,Linker2012,Ludi2012,Spiegel2007}. 
  In our data, we found that among patients whose event times were censored,
  59.8\% of them were without any of these 10 conditions identified by the
  incidence part of our model. As such, these patients could be regarded as with
  the least risk of having subsequent suicide attempt. Our findings could help
  clinical practice for a better allocation of the limited resource for suicide
  prevention.

Our proposed method selected drug dependence (ICD-9 304) and psychosocial
diagnoses (ICD-9 V62) in both the incidence part and
the latency part, which suggested that these two conditions not only increased
the probability of being susceptible to subsequent suicide attempt after the
initial hospitalization but also associated with a reduced time to subsequent
attempt. On the other hand, the other three naive methods all missed the selection of drug
dependence. The link between drug dependence and suicide attempt
is well known in children and adolescents \citep{Berman1990, DoshiChen2020, LuoChen2022}, and suicide death
is recognized as a major component of the ongoing opioid crisis by the National Institute
on Drug Abuse. 


We also attempted to conduct post-selection inference on the I.Cure model. To
this end, we simply refitted the selected model without regularization and performed bootstrap to
obtain 95\% confidence intervals of the coefficient estimates. According to \citet{zhao2017defense}, such a naive two-step procedure could
still yield asymptotically valid inference under certain conditions. The results
are shown in 
Supplementary Materials.
Indeed, most of the predictors were significant except ICD-9 298, 300, and V62
in the incidence part.\\

\noindent\textbf{Prediction performance}

We evaluated the model prediction performance based on the selected
variables from those different approaches by a random splitting procedure.  To be more
specific, we randomly split subjects in Case~1--2, respectively, into a training
set with probability 0.6 and a test set with probability 0.4.  The number of
subjects in Case~1 and Case~2 in the testing set is thus 294 and 2,657,
respectively.  The proposed method was fitted to the split training set and
Case~3, while for the Cure1 method, only the split training set was utilized.
By definition, the Cure2 method and the Cure3 method took uncertain events in
Case~3 as censoring and actual events, respectively.  
The procedure was repeated 1,000 times.

We applied the time-varying sensitivity and specificity estimator proposed by
\citet{uno2007jasa} to evaluate the prediction performance based on the
estimated survival probabilities over the test set that consists of Case~1--2
only.  We focused on short-term survival in the first 4 months when about 40\%
of subjects in Case~3 had observed time.  Motivated by the clinical setting, in which suicide related interventions would
not be universally implemented but be targeted on those at highest risk,
we evaluated the predictive performance within an estimated high-risk group of small size.
\begin{figure}
  \centering
  \includegraphics[width=\linewidth]{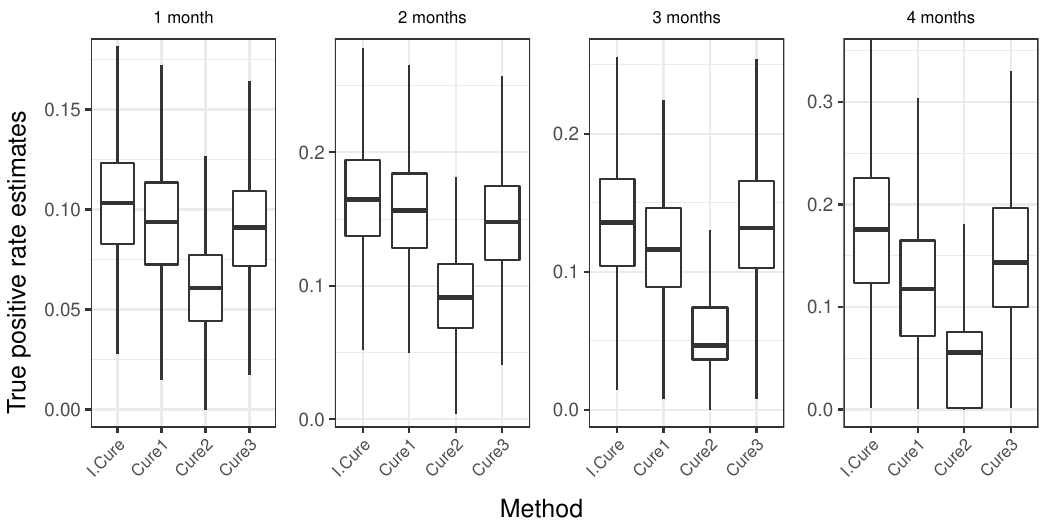}
  \caption[The boxplots of true positive rate.]%
  {Suicide risk study: boxplot matrix of true positive rate of predicting 1--4 months survival
    with the proportion of high-risk group controlled at 5\%.}%
  \label{fig:data:tpr}
\end{figure}
Figure~\ref{fig:data:tpr} provides a visual comparison of the prediction
performance for 1--4 month's survival when the proportion of the high risk group
is controlled at 5\% of the population.  The TPR from the I.Cure method was
the largest in the first four months compared with those competing methods.  The
Cure1 and Cure3 method provided the second largest TPR in the first two month
and last two months, respectively, while the Cure2 method always gave the worst
performance in the first four months.



\noindent\textbf{Understanding suspected attempts}

We then investigated the suspected attempts in Case~3, to understand which
injuries of undetermined intent were more likely to be self-inflicted. The estimated
posterior probabilities were used to classify whether these patients had
actual attempts, censored events, or cured using the Bayes' rule. Among those
173 patients in Case~3 with injury of undetermined origin, 36 (20.8\%) patients
were classified as having an actual events (determined suicide attempts), 137
(79.1\%) patients were classified as being cured, and none of the patients were
classified as being censored under risk. The result was consistent with our
early conclusion on the sufficient follow-up period.

\begin{figure}[tbp]
  \centering
  \begin{minipage}[t]{.48\textwidth}
    \centering
    \includegraphics[width=0.95\linewidth]{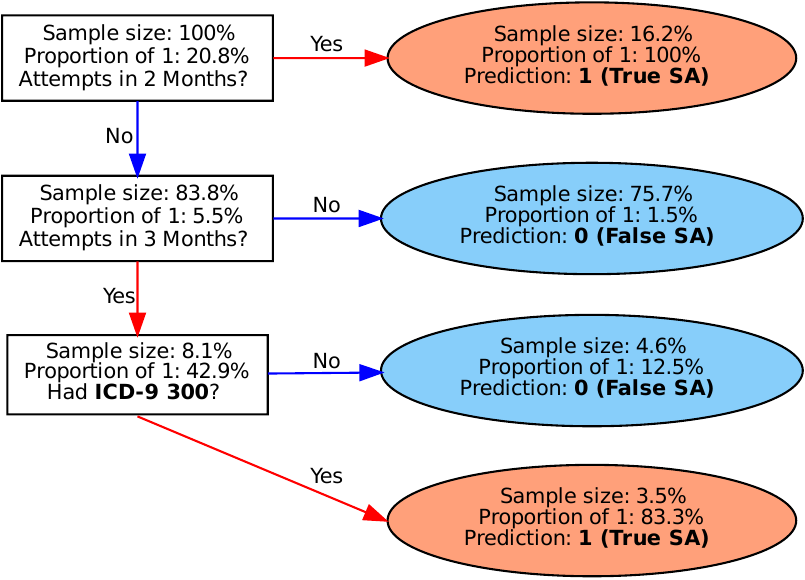}
    (a) Model with survival times
  \end{minipage}%
  \begin{minipage}[t]{.48\textwidth}
    \centering
    \includegraphics[width=0.95\linewidth]{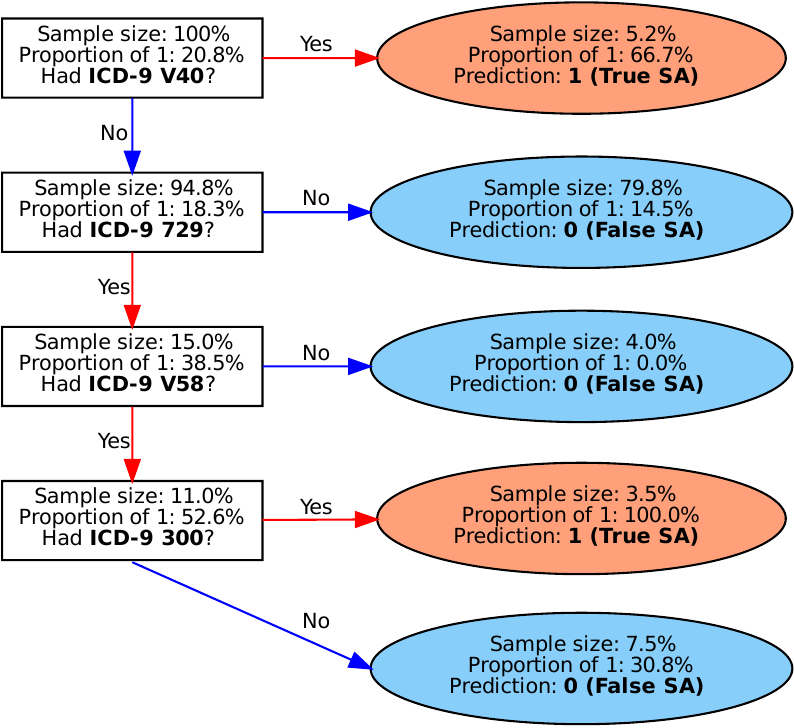}
    (b) Model without survival times
  \end{minipage}
  \caption[Classification Tree of Case 3a among Case 3.]%
  {Suicide risk study: classification trees of actual suicide attempts (SA) among subjects having
    undetermined SA based on identified Case 3a among Case 3 from I.Cure.}%
  \label{fig:data:tree-case3a}
\end{figure}


Given the identified subjects in Case~3a from I.Cure, we fitted two classification
tree models with a regular complexity pruning procedure among subjects in Case~3
to further explore ICD-9 categories that predict Case~3a.
Figure~\ref{fig:data:tree-case3a} provides their visualization.

In tree model (a), we included the survival times of the suspected suicide
attempt as one of the predictors. The results show that the patients
who have had undetermined suicide attempts within 2 months
(accounting for 16.2\% of patients) after the initial
hospitalization were very likely to have made actual suicide attempts.
On the other hand, those patients who have not made a
suspected attempt within 3 months (accounting for 75.7\% of patients)
were unlikely to have made actual suicide attempts.  Among the
remaining (8.1\%) patients, those having anxiety, dissociative, and
somatoform disorders (ICD-9 300) were predicted to have had an actual
suicide attempt.
Then, in tree model (b), we excluded the survival times to see how
well we could use historical information to predict Case~3a.
In the root node, patients having mental and behavioral problems (ICD-9 V40)
were predicted to have had actual suicide attempts.  In the second
node, the remaining (94.8\%) patients without disorders of soft tissues (ICD-9
729) were predicted to not have had actual suicide attempts.
About 85\% of patients were classified to the corresponding terminal nodes by
the first two nodes.
For the remaining (15.0\%) patients, those who required continued
care after initial treatment (ICD-9 V58) and had anxiety, dissociative, and
somatoform disorders (ICD-9 300) were predicted to have had actual attempts.

Out of 36 identified cases, model (a) was able to capture
34 of them, while model (b) only captured 15, which suggested that the
estimated baseline hazard functions from I.Cure based on observations in Case~1
and Case~2 played an important role in identifying Case~3a.
Our analysis provided the compositions of these suspected suicide events, which
could be insightful and further examined by domain experts for a more effective
suicide prevention strategy among youth and young adults.







Another interesting question is that, among all the actual attempts, why some
were coded as determined while some were coded as suspected? That is, how to
distinguish the subjects in Case~1 and the identified subjects in Case~3a? We fitted a classification tree model for these two sets of subjects. 
The fitted model suggested that subjects having prior poisoning by unspecified
biological substances (ICD-9 979), or general symptoms (ICD-9 780), such as
physiological exhaustion (ICD-9 780.79), unspecified altered mental status
(ICD-9 780.97) were more likely to be recorded as having suspected suicide
attempts with E98 codes.

\section{Discussion}%
\label{sec:discussion}

We have examined a general modeling setup for survival data with the presence of a
cure fraction and uncertain events indicators.  In our application, we obtained
insightful results on potential risk factors associated with both the occurrence
and timing of subsequent suicide attempt following initial suicide-related
hospitalization among adolescents.
Several important risk factors might have been missed by the
other naive approaches that improperly treat uncertain diagnosis records. We
thoroughly examined the estimated cure and event status among those subjects with
suspected suicide attempt diagnosis, which provided important insights on how to
identify the true suicidal events among them and how to utilize such uncertain
records in predictive modeling of suicide.


Several directions are worth pursuing for future research.  For the proposed
method, the relative risk on the conditional latency is specified by a Cox
proportional hazards model, which may be restrictive. It is interesting to consider testing procedures of the proportional hazard assumption in the presence of outcome uncertainty and cure fraction \citep{wileyto2013assessing}, and develop other more
flexible models to relax the proportional hazards assumption.  In
the estimation procedure, we applied in the
E-steps the zero tail completion, which essentially treats the subjects who have
survived after the largest uncensored event time as not susceptible. This may
be too restricted in some applications. It is worth exploring other
methods such as exponential
and Weibull tail completion \citep{peng2003csda-est}
to model the tail of the conditional survival curve.
In fact, we tried exponential tail completion in our simulation
studies in Section~\ref{sec:simu} and obtained very similar results from using
zero tail completion.

Given the prevalence of uncertainty
in the identification of medical conditions from diagnosis codes, we believe
our approach can be broadly applied to improve many studies on health
conditions with real-world claims data.




\section*{Acknowledgment}

Chen's work was partially supported U.S. National Science Foundation (DMS-1613295, IIS-1718798) and U.S. National Institutes of Health (R01-MH124740).

\section*{Conflict of Interest}

None.


\clearpage

\newcommand{\beginsupplement}{%
  \setcounter{table}{0}
  \renewcommand{\thetable}{S\arabic{table}}%
  \setcounter{figure}{0}
  \renewcommand{\thefigure}{S\arabic{figure}}%
  \setcounter{section}{0}
  \renewcommand{\thesection}{\Alph{section}}%
  \setcounter{subsection}{0}
  \renewcommand{\thesubsection}{\Alph{section}.\arabic{subsection}}%
}
\begin{center}
  \section*{\Huge Supplementary Material}
\end{center}
\beginsupplement

\section{Derivation of Likelihood Functions}

To simplify the notations, we let $p(\cdot)$ be a generic notation representing
the probability density function.
For discrete random variables, we can still view $p(\cdot)$ as the density
function by the adoption of counting measure.

We derive both the complete-data likelihood and the corresponding observed-data likelihood. Firstly, for subject $j \in I_1$, we observe that $T_j=V_j$ and
$\sigma_j=\delta_j=Z_j=\xi_j=1$; no label is missing.
Thus, the likelihood contribution can be derived as follows:
\begin{align*}
  & p(T_j = t_j, \Delta_j = 1, \xi_j = 1, Z_j = 1)\\
  =~ & p(\xi_j = 1, T_j = t_j, \Delta_j = 1 \mid Z_j = 1)p(Z_j = 1)\\
  =~ & p(\xi_j = 1 \mid T_j = t_j, \Delta_j = 1, Z_j = 1)
    p(T_j = t_j, \Delta_j = 1 \mid Z_j = 1)p(Z_j = 1)\\
  =~ & p(\xi_j = 1\mid T_j = t_j)p(T_j = t_j, \Delta_j = 1 \mid Z_j = 1)p(Z_j = 1)\\
  =~ & q_j p_j f_j(t_j)\overline{G}(t_j),
\end{align*}%
where $q_j=p(\xi_j = 1\mid T_j = t_j)$,
$f_j(t)$ is the density function of $V_j$
given that the subject $j$ is susceptible,
and $\overline{G}(t)$ is the survival function of $C_j$.

Secondly, for subject $j \in I_2$, we observe $T_j = C_j$,
$\Delta_j = 0$, and $\xi_j = 1$.
But $Z_j$ is not observed.
We have
\begin{align*}
  & p(T_j = t_j, \Delta_j = 0, \xi_j = 1 \mid Z_j = 1)\\
  =~ & p(\xi_j = 1 \mid T_j = t_j, \Delta_j = 0, Z_j = 1)
       p(T_j = t_j, \Delta_j = 0 \mid Z_j = 1)\\
  =~ & p(\xi_j = 1 \mid T_j = t_j)
       p(C_j = t_j, V_j > C_j \mid Z_j = 1)\\
  =~& q_j g(t_j) S_j(t_j),\\
  & p(T_j = t_j, \Delta_j = 0, \xi_j = 1 \mid Z_j = 0)\\
  =~& p(\xi_j = 1 \mid T_j = t_j, \Delta_j = 0, Z_j = 0)
      p(T_j = t_j, \Delta_j = 0\mid Z_j = 0)\\
  =~ & p(\xi_j = 1 \mid T_j = t_j)
       p(C_j = t_j, V_j > C_j \mid Z_j = 0)\\
  =~ & q_j g(t_j),
\end{align*}%
where $g(t)$ is the density function of $C_j$,
$S_j(t)$ is the survival function of $V_j$
given that the subject $j$ is susceptible.
Therefore, the observed-data likelihood contribution of subject $j$ is
\begin{align*}
  p(T_j = t_j, \Delta_j = 0, \xi_j = 1)
  = & \sum_{z_j \in \{0, 1\}}
      p(T_j = t_j, \Delta_j = 0, \xi_j = 1 \mid Z_j = z_j)
      p(Z_j = z_j)\\
  =~ & q_j g(t_j) [p_j S_j(t_j) + (1 - p_j)],
\end{align*}
and the likelihood contribution under the complete data is
\begin{align*}
  p(T_j = t_j, \Delta_j = 0, \xi_j = 1, Z_j = z_j) =
  q_j g(t) {\left[p_j S_j(t_j)\right]}^{z_j}
  {\left[(1 - p_j) \right]}^{1 - z_j}.
\end{align*}

Lastly, for subject $j \in I_3$, we only observe uncertain event time
$T_j=\min(V_j, C_j)$ with $\xi_j = 0$.
Neither $\Delta_j$ nor $Z_j$ is observable.
We can obtain the likelihood contribution of subject $j$ in Case 3a and Case 3b
under the observed data as follows:
\begin{align*}
  & p(T_j = t_j, \xi_j = 0 \mid Z_j = 1)\\
  =~ & p(T_j = t_j, \xi_j = 0, \delta_j = 0 \mid Z_j = 1) +
       p(T_j = t_j, \xi_j = 0, \delta_j = 1 \mid Z_j = 1)\\
  =~ & p(\xi_j = 0 \mid T_j = t_j, \Delta_j = 1, Z_j = 1)
       p(T_j = t_j, \Delta_j = 1 \mid Z_j = 1) +\\
    ~& p(\xi_j = 0 \mid T_j = t_j, \Delta_j = 0, Z_j = 1)
       p(T_j = t_j, \Delta_j = 0 \mid Z_j = 1)\\
  =~ & p(\xi_j = 0 \mid T_j = t_j)
       \left[p(V_j = t_j, C_j \ge V_j\mid Z_j = 1) +
       p(C_j = t_j, V_j > C_j\mid Z_j = 1)\right]\\
  =~ & (1 - q_j)\left[f_j(t_j) \overline{G}(t_j) + g(t_j) S_j(t_j)\right],
\end{align*}
and thus
\begin{align*}
  p(T_j = t_j, \xi_j = 0, Z_j = 1)
  = & p(T_j = t_j, \xi_j = 0\mid Z_j = 1)p(Z_j=1)\\
  = & (1 - q_j)p_j\left[f_j(t_j) \overline{G}(t_j) + g(t_j) S_j(t_j)\right].
\end{align*}
Similarly, as for the subject $j$ in Case 3c, we can derive the observed data
likelihood contribution as follows:
\begin{align*}
  & p(T_j = t_j, \xi_j = 0 \mid Z_j = 0)\\
  =~ & p(T_j = t_j, \xi_j = 0, \delta_j = 0 \mid Z_j = 0)\\
  =~ & p(\xi_j = 0 \mid T_j = t_j, \Delta_j = 0, Z_j = 0)
       p(T_j = t_j, \Delta_j = 0 \mid Z_j = 0)\\
  =~ & p(\xi_j = 0 \mid T_j = t_j)
       p(C_j = t_j, V_j > C_j\mid Z_j = 1)\\
  =~ & (1 - q_j) g(t_j),
\end{align*}
and thus
\begin{align*}
  p(T_j = t_j, \xi_j = 0, Z_j = 0)
  = & p(T_j = t_j, \xi_j = 0 \mid Z_j = 0) p(Z_j = 0)\\
  = & (1 - q_j) g(t_j) (1 - p_j).
\end{align*}
Notice that $\Pr(\Delta_j = 1 \mid Z_j = 0) = 0$.
Therefore, the observed data likelihood contribution for subject $j$ in Case 3
is
\begin{align*}
  p(T_j = t_j, \xi_j = 0)
  =~ & (T_j = t_j, \xi_j = 0, Z_j = 1) + p(T_j = t_j, \xi_j = 0, Z_j = 0)\\
  =~ & (1 - q_j)p_j\left[f_j(t_j) \overline{G}(t_j) + g(t_j) S_j(t_j)\right] +
       (1 - q_j) g(t_j) (1 - p_j),
\end{align*}
and the likelihood contribution under the complete data is
\begin{align*}
  & p(T_j = t_j, \Delta_j = \delta_j, \xi_j = 0, Z_j = z_j)\\
  =~ & (1 - q_j)
      {\left[p_j f_j(t_j) \overline{G}(t_j)\right]}^{\delta_j z_j}
      {\left[p_j g(t_j) S_j(t_j)\right]}^{(1 - \delta_j) z_j}
      {\left[(1 - p_j) g(t_j)\right]}^{1 - z_j}.
\end{align*}

From the derived likelihood contributions for Case 1--3, we write down the
likelihood functions under the complete data as follows:
\begin{align*}
  & \prod_{j \in I_1} q_j p_j f_j(t_j)\overline{G}(t_j)
  \prod_{j \in I_2} q_j g(t) {\left[p_j S_j(t_j)\right]}^{z_j}
  {\left[(1 - p_j) \right]}^{1 - z_j}\\
  & \prod_{j \in I_3} (1 - q_j)
    {\left[p_j f_j(t_j) \overline{G}(t_j)\right]}^{\delta_j z_j}
    {\left[p_j g(t_j) S_j(t_j)\right]}^{(1 - \delta_j) z_j}
    {\left[(1 - p_j) g(t_j)\right]}^{1 - z_j}.
\end{align*}
In addition, the missing at random assumption of the $\{\xi_j\}$ suggests that
the likelihood contributions of $\{\xi_j\}$ do not involve $\{\delta_j\}$ or
$\{z_j\}$, and thus can be treated as nuisance components that play no role in
estimating the parameters of interest.  As a result, the derived EM algorithm
for those parameters of interest in next section do not depend on the ${q_j}$ at
all.

Therefore, the simplified complete-data likelihood function and the
corresponding observed-data likelihood function are, respectively, as follows:
\begin{align*}
  & \prod_{j \in I_1} p_j f_j(t_j)\overline{G}(t_j)
  \prod_{j \in I_2} g(t) {\left[p_j S_j(t_j)\right]}^{z_j}
  {\left[(1 - p_j) \right]}^{1 - z_j}\\
  & \prod_{j \in I_3}
    {\left[p_j f_j(t_j) \overline{G}(t_j)\right]}^{\delta_j z_j}
    {\left[p_j g(t_j) S_j(t_j)\right]}^{(1 - \delta_j) z_j}
    {\left[(1 - p_j) g(t_j)\right]}^{1 - z_j},
\end{align*}
and
\begin{align*}
  & \prod_{j \in I_1} p_j f_j(t_j)\overline{G}(t_j)
  \prod_{j \in I_2} g(t) {\left[p_j S_j(t_j) + (1 - p_j)\right]}\\
  & \prod_{j \in I_3}
    {\left[p_j f_j(t_j) \overline{G}(t_j) + p_j g(t_j) S_j(t_j)
    (1 - p_j) g(t_j)\right]}.
\end{align*}

\section{Derivation of Model Estimation Procedure}%
\label{sec:deriv-em}

\subsection{Derivation of the E-step}\label{sec:deriv-EM:E}

There is no missing indicators for subjects in Case~1.
So we only need to derive the E-step for subjects in Case~2 and Case~3,
respectively.
For subjects in Case~2, the log-likelihood function under the complete data
$\bm{Y}_2$ is
\begin{align*}
  & \ell(\bm{\theta} \mid \bm{Y}_2)\\
  = & \sum_{j \in I_2} \log{g(t_j)} +
    z_j \left(\log{p_j} + \log{S_j(t_j\mid Z_j = 1)}\right) +
      (1 - z_j) \log{(1 - p_j)}\\
  = & \sum_{j \in I_2}
      z_j \left[(\bx_j^{\top}\bm{\gamma} + \gamma_0) -
      H_0(t_j\mid Z_j = 1)\exp(x_j^{\top} \bb)\right]
      - \log\left[1 + \exp(\bx_j^{\top}\bm{\gamma} + \gamma_0)\right]\\
  & + \log h_c(t_j) - H_c(t_j).
\end{align*}%

Let $\bm{\theta}^{(i)}$ denote the set of parameter estimates at the $i$-th
iteration.  Given the observed data of subject $j$ in Case~2 and
$\bm{\theta}^{(i)}$, by Bayes rule we have
\begin{align*}
  \E(Z_j \mid T_j = t_j, \Delta_j = 0, \xi_j = 1, \bm{\theta}^{(i)})
  & = \Pr(Z_j = 1 \mid T_j = t_j, \Delta_j = 0, \xi_j = 1, \bm{\theta}^{(i)})\\
  & = \frac{\Pr(Z_j = 1, T_j = t_j, \Delta_j = 0, \xi_j = 1 \mid \bm{\theta}^{(i)})}
    {\sum_{z_j \in \{0, 1\}}
    p(Z_j = z_j, T_j = t_j, \Delta_j = 0, \xi_j = 1 \mid \bm{\theta}^{(i)})}\\
  & = \frac{p_j^{(i)} S_j^{(i)}}{p_j^{(i)} S_j^{(i)} + (1 - p_j^{(i)})}.
\end{align*}%
where $p_j^{(i)}$ and $S_j^{(i)}$ are the estimate of $p_j$ and
$S_j(t_j\mid Z_j = 1)$ at the $i$-th iteration, respectively.
Let $\bm{Y}^{O}_{2}$ denote the observed data for subjects in Case~2.
The E-step for subjects in Case~2 is as follows:
\begin{align*}
  & \E \ell(\bm{\theta} \mid \bm{Y}^{O}_{2}, \bm{\theta}^{(i)})\\
  = & \sum_{j \in I_2}
      v_j^{(i)} \left[(\bx_j^{\top}\gamma + \gamma_0) -
      H_0(t_j\mid Z_j = 1)\exp(x_j^{\top} \bb)\right] - \log\left[
    1 + \exp(\bx_j^{\top}\gamma + \gamma_0)\right] \nonumber \\
  & + \log h_c(t_j) - H_c(t_j),
\end{align*}%
where
$v_j^{(i)} = p_j^{(i)} S_j^{(i)} / [p_j^{(i)} S_j^{(i)} + (1 - p_j^{(i)})]$.

Similarly, for subjects in Case~3, the log-likelihood function under the
complete data is
\begin{align*}
  & \ell(\bm{\theta} \mid \bm{Y}_3)\\
  = & \sum_{j \in I_3} \delta_j z_j \left[\log p_j
      + \log f_j(t_j\mid Z_j = 1) + \log \overline{G}(t_j)\right]\\
    & + (1 - \delta_j) z_j \left[\log p_j +
      \log g(t_j) + \log S_j(t_j\mid Z_j = 1)\right]
      + (1 - z_j) \left[\log (1 - p_j) + \log g(t_j)\right].
\end{align*}%
Define $m_{j,k}^{(i)}$, $k\in\{1,2,3\}$, as follows:
\begin{align*}
  m_{j,1}^{(i)} :=
  p(\Delta_j = 1, Z_j = 1, T_j = t_j,\xi_j = 0 \mid \bm{\theta}^{(i)})
    & = q_j^{(i)} p_j^{(i)} h_j^{(i)} S_j^{(i)} \overline{G}_j^{(i)},\\
  m_{j,2}^{(i)} :=
  p(\Delta_j = 0, Z_j = 1, T_j = t_j,\xi_j = 0 \mid \bm{\theta}^{(i)})
    & = q_j^{(i)} p_j^{(i)} S_j^{(i)} g_j^{(i)},\\
  m_{j,3}^{(i)} :=
  p(\Delta_j = 0, Z_j = 0, T_j = t_j,\xi_j = 0 \mid \bm{\theta}^{(i)})
    & = q_j^{(i)} (1 - p_j^{(i)}) g_j^{(i)},
\end{align*}
where $h_j^{(i)}$, $g_j^{(i)}$, and $\overline{G}_j^{(i)}$ are the
estimate of $h_j(t_j\mid Z_j = 1)$, $g(t_j)$, and $\overline{G}_j(t_j)$ at the
$i$-th iteration, respectively.
Let $m_{j}^{(i)} = \sum_{k=1}^{3} m_{j,k}^{(i)}$ and define
$w_{j,k}^{(i)}$, $k\in\{1,2,3\}$, as follows:
\begin{align*}
  w_{j,1}^{(i)}
  & := 
    \E(\Delta_j Z_j \mid T_j = t_j, \xi_j =0, \bm{\theta}^{(i)})
    = m_{j,1}^{(i)} / m_{j}^{(i)},\\
  w_{j,2}^{(i)}
  & : =
    \E \left[(1 - \Delta_j) Z_j \mid T_j = t_j, \xi_j = 0, \bm{\theta}^{(i)}\right]
    = m_{j,2}^{(i)} / m_{j}^{(i)},\\
  w_{j,3}^{(i)}
  & :=
    \E (1 - Z_j \mid T_j = t_j, \xi_j = 0, \bm{\theta}^{(i)})
    = m_{j,3}^{(i)} / m_{j}^{(i)}.
\end{align*}%
The E-step for subjects in Case~3 is then given below.
\begin{align*}
  & \E \ell(\bm{\theta} \mid \bm{Y}^{O}_3, \bm{\theta}^{(i)})\\
  = & \sum_{j \in I_3} w_{j,1}^{(i)} \left[\log h_0(t_j\mid Z_j=1) + \bx^{\top}_j \bb
      \right] + \left(w_{j,1}^{(i)} + w_{j,2}^{(i)}\right)
    \left[ \bx_j^{\top} \gamma + \gamma_0 - H_0(t\mid Z_j=1)
    \exp(\bx_j^{\top}\bb) \right]\\
  & - \log\left[1 + \exp(\bx_j^{\top} \gamma + \gamma_0)\right] +
    \left(w_{j,2}^{(i)} + w_{j,3}^{(i)}\right)\log h_c(t_j) - H_c(t_j).
\end{align*}

\subsection{Derivation of the M-step}\label{sec:deriv-EM:M}

The conditional expectation of the complete-data log-likelihood given the
observed data and estimates from last step can be decomposed into several parts
that involve exclusive sets of parameters.

First of all, the part involving $h_c(\cdot)$ is
\begin{align}\label{eqn:hclike}
  \E \ell(h_c(\cdot) \mid \bm{Y}^{O}, \bth^{(i)}) =
  \sum_{j=1}^n -H_c(t_j) + \sum_{j\in I_2} \log h_c(t_j) +
  \sum_{j\in I_3} (w_{j,2}^{(i)} + w_{j,3}^{(i)}) \log h_c(t_j).
\end{align}
Because $H_c(t)$ is a non-decreasing function, $h_c(t)$
maximizing~\eqref{eqn:hclike} is a discrete function that takes non-zero
values at $t_j$, $j\in I_2 \cup I_3$.  Let $Y_j(t) = \bm{1}(t \le t_j)$ denote
the at-risk indicator. Define that
$C_{j,2}(t) = \bm{1}(t_j \le t) \bm{1}(j \in I_2)$,
$C_{j,3}(t; \bth^{(i)}) = (w_{j,2}^{(i)} + w_{j,3}^{(i)}) \bm{1}(t_j\le
t)\bm{1}(j\in I_3)$, and
$C(t; \bth^{(i)}) = \sum_{j=1}^n C_{j,2}(t) + C_{j,3}(t; \bth^{(i)}) =
\sum_{j\in I_2} C_{j,2}(t) + \sum_{j\in I_3} C_{j,3}(t; \bth^{(i)})$,
and let $\dif C(t; \bth^{(i)}) = C(t; \bth^{(i)}) - C(t-0; \bth^{(i)})$
denote the jump size of $C(t; \bth^{(i)})$.
Then we may rewrite~\eqref{eqn:hclike} to allow ties in $t_j$,
$j\in I_2 \cup I_3$ as
\begin{align}\label{eqn:hclike-tie}
  \E \ell(h_c(\cdot) \mid \bm{Y}^{O}, \bth^{(i)}) =
  \sum_{t \in \mathcal{T}_2 \cup \mathcal{T}_3}
  \left[ \dif C(t; \bth^{(i)}) \log h_c(t) -
  h_c(t) \sum_{j=1}^n Y_j(t)\right],
\end{align}
where $\mathcal{T}_k = \{t_j \mid j \in I_k\}$.
Maximizing~\eqref{eqn:hclike-tie} gives that
\begin{align*}
  h_c^{(i+1)}(t) = \frac{\dif C(t; \bth^{(i)})}{\sum_{j=1}^n Y_j(t)}.
\end{align*}%

Secondly, the part involving both $h_0(\cdot\mid Z_j=1)$ and $\bb$ is
\begin{align}
  & \E \ell(\bb, h_0(\cdot\mid Z_j=1) \mid \bm{Y}^{O}, \bth^{(i)})
    \nonumber\\
  = & ~ \sum_{j\in I_1} \bx_{j}^{\top}\bb +
      \log h_0(t_j\mid Z_j=1) - H_0(t_j\mid Z_j=1) \exp(\bx_{j}^{\top} \bb)
      \nonumber\\
  & + \sum_{j\in I_2} v_j^{(i)}
    \left[- H_0(t_j\mid Z_j=1) \exp(\bx_{j}^{\top} \bb)\right]
    \nonumber \\
  & + \sum_{j\in I_3} w_{j,1}^{(i)} \left[\bx_{j}^{\top} \bb +
    \log h_0(t_j\mid Z_j=1)\right] +
    \left(w_{j,1}^{(i)} + w_{j,2}^{(i)}\right)
    \left[- H_0(t_j\mid Z_j=1) \exp(\bx_{j}^{\top} \bb)\right]\label{eqn:h0like}.
\end{align}
Similarly, $h_0(t\mid Z_j=1)$ maximizing~\eqref{eqn:h0like} is a discrete
function that takes non-zero values at $t_j$, $j\in I_1 \cup I_3$.  Define that
$M_j^{(i)} = \bm{1}(j \in I_1) + v_j^{(i)}\bm{1}(j\in I_2) + (w_{j,1}^{(i)} +
w_{j,2}^{(i)}) \bm{1}(j \in I_3)$,
$N_j^{(i)} = \bm{1}(j \in I_1) + w_{j,1}^{(i)}\bm{1}(j\in I_3)$,
$N_j(t; \bth^{(i)}) = \bm{1}(t_j\le t) N_j^{(i)}$, and
$N(t; \bth^{(i)}) = \sum_{j=1}^n N_j(t; \bth^{(i)})$.  Let
$\dif N(t; \bth^{(i)}) = N(t; \bth^{(i)}) - N(t-0; \bth^{(i)})$
denote the jump size of $N(t; \bth^{(i)})$.  Then we may
rewrite~\eqref{eqn:h0like} to allow ties in $t_j$, $j\in I_1 \cup I_3$ as
\begin{align}\label{eqn:h0like-tie}
  & \E \ell(\bb, h_0(\cdot\mid Z_j=1) \mid \bm{Y}^{O}, \bth^{(i)})
    \nonumber\\
  = & \sum_{t\in \mathcal{T}_1 \cup \mathcal{T}_3}
      \dif N(t; \bth^{(i)}) \log h_0(t\mid Z_j=1) - h_0(t\mid Z_j=1)
      \left[\sum_{j=1}^n Y_j(t) M_j^{(i)}\exp(\bx_{j}^{\top}\bb) \right]
      \nonumber\\
  & + \sum_{j=1}^n N_{j}^{(i)} \bx_{j}^{\top}\bb.
\end{align}
Maximizing~\eqref{eqn:h0like-tie} with regards to $h_0(t\mid Z_j=1)$ gives
\begin{align}
  h_0^{(i+1)}(t\mid Z_j=1) = \frac{\dif N(t; \bth^{(i)})}{
  \sum_{j=1}^n Y_j(t) M_j^{(i)} \exp(\bx_{j}^{\top} \bb)},\label{eqn:h0est}
\end{align}
which is similar to the ``Breslow estimator'' \citep{breslow1974biometrics} for
the regular Cox model.
Plugging~\eqref{eqn:h0est} into~\eqref{eqn:h0like-tie} gives the profiled
expectation,
\begin{align*}
  \E \ell(\bb, h_0^{(i+1)}(\cdot\mid Z_j=1) \mid \bm{Y}^{O}, \bth^{(i)})
  = \sum_{t\in \mathcal{T}_1 \cup \mathcal{T}_3} \dif N(t; \bth^{(i)})
  \left[\log \dif N(t; \bth^{(i)}) - 1\right] + \ell(\bb \mid \bth^{(i)}),
\end{align*}%
where
\begin{align}\label{eqn:cox:m}
  \ell(\bb \mid \bth^{(i)})
  & = \sum_{j=1}^n \int_0^{\infty} I_j(\bb, t \mid \bth^{(i)}) \dif N_j(t; \bth^{(i)}),\\
  I_j(\bb, t \mid \bth^{(i)})
  & = \bx_j^{\top} \bb -
    \log \left[\sum_{j=1}^n Y_j(t) M_j^{(i)} \exp(\bx_{j}^{\top} \bb)\right].
    \nonumber
\end{align}%
This profiling approach is similar to the partial likelihood of
\citet{cox1975biometrika} except that the susceptible indicators and the
distribution of the censoring time come into play through $M_{j}^{(i)}$ and
$\dif N_{j}(t; \bm{\theta}^{(i)})$.

Thirdly, the part involving $\bg$ and $\gamma_0$ is
\begin{align}\label{eqn:cure:m}
  \E \ell(\bg, \gamma_0 \mid \bm{Y}^{O}, \bth^{(i)}) =
  \sum_{j=1}^n (\bx_j^{\top} \bg + \gamma_0) M_j^{(i)}
  - \log \left[1 + \exp(\bx_j^{\top} \bg + \gamma_0)\right],
\end{align}%
which is similar to the log-likelihood function of a regular logistic regression
model.
For ease of notation, we henceforth
denote~\eqref{eqn:cure:m} by $\ell(\bg, \gamma_0 \mid \bth^{(i)})$.

Without regularization, maximizing~\eqref{eqn:cox:m} with respect to $\bb$ and
maximizing~\eqref{eqn:cure:m} with respect to $\bg$ and
$\gamma_0$ can be done by the Newton-Raphson
algorithm or its variants, in a similar fashion to the Cox model and the
logistic regression model, respectively.
In our implementation, we adopted the monotonic quadratic
approximation algorithm proposed by \citet{bohning1988aism} due to its
monotonic convergence property.

With regularization, the loss function involving $\bb$, $\bg$, and $\gamma_0$
that we aim to minimize are, respectively,
\begin{align}
  J(\bb \mid \bth^{(i)})
  & = -\frac{1}{n} \ell(\bb \mid \bth^{(i)}) +
    P_{1}(\bb; \alpha_1, \lambda_1),\label{eqn:cox:reg}\\
  J(\bg, \gamma_0 \mid \bth^{(i)})
  & = -\frac{1}{n} \ell(\bg, \gamma_0 \mid \bth^{(i)}) +
    P_{2}(\bg; \alpha_2, \lambda_2).\label{eqn:cure:reg}
\end{align}
The M-steps for $h_c(\cdot)$ and $h_0(\cdot\mid Z_j=1)$ remains the same.
We derived the CMD algorithms for~\eqref{eqn:cox:reg} and~\eqref{eqn:cure:reg},
respectively, in the following paragraphs.

The loss function of $\beta_k$, $k\in\{1,\ldots,p\}$ for fixed tuning parameter
$\alpha_1$, $\lambda_1$, and fixed $\beta_l = \tilde{\beta}_l$, $l \neq k$, is
\begin{align}
  J_k(\beta_k\mid \bth^{(i)}) =
  - \frac{1}{n} \ell(\beta_k \mid \beta_{l} = \tilde{\beta}_{l},~
  l \neq k,~\bth^{(i)}) + \lambda_1 \left(\alpha_1 \omega_k \lvert \beta_k \rvert
  + \frac{1 - \alpha_1}{2} \beta_k^2\right),\label{eq:loss-beta-k}
\end{align}
which is an univariate function.  However, there is no a closed-form solution to
minimize~\eqref{eq:loss-beta-k} efficiently.
Thus, we derived an update step of $\beta_k$ to decrease~\eqref{eq:loss-beta-k}
by following the CMD algorithm.  Let
$\dot{\ell}_k(\bb \mid \bth^{(i)})$ and $\ddot{\ell}_k(\bb \mid \bth^{(i)})$
denote the first and second partial derivative of $\ell(\bb \mid \bth^{(i)})$
with respect to $\beta_k$, respectively.  We have
\begin{align*}
  \dot{\ell}_k(\bb \mid \bth^{(i)})
  & = \sum_{j=1}^n \int_0^{\infty}
    \pdv{I_j(\bb, t \mid \bth^{(i)})}{\beta_k} \dif N_j(t; \bth^{(i)}),\\
  \pdv{I_j(\bb, t \mid \bth^{(i)})}{\beta_k}
  & = x_{j,k} - \frac{\sum_{j=1}^n x_{j,k} Y_j(t) M_j^{(i)} \exp(\bx_{j}^{\top} \bb)}
    {\sum_{j=1}^n Y_j(t) M_j^{(i)} \exp(\bx_{j}^{\top} \bb)},
\end{align*}%
and
\begin{align*}
  \ddot{\ell}_k(\bb \mid \bth^{(i)})
  & = \sum_{j=1}^n \int_0^{\infty}
    \pdv[2]{I_j(\bb, t \mid \bth^{(i)})}{\beta_k} \dif N_j(t; \bth^{(i)}),
    \nonumber\\
  \pdv[2]{I_j(\bb, t \mid \bth^{(i)})}{\beta_k}
  & = - \frac{\sum_{j\in \mathcal{R}(t)}
    x_{j,k}^2 M_j^{(i)} \exp(\bx_{j}^{\top} \bb)}
    {\sum_{j\in \mathcal{R}(t)}  M_j^{(i)} \exp(\bx_{j}^{\top} \bb)} +
    {\left[\frac{\sum_{j\in \mathcal{R}(t)} x_{j,k} M_j^{(i)} \exp(\bx_{j}^{\top} \bb)}
    {\sum_{j\in \mathcal{R}(t)} M_j^{(i)} \exp(\bx_{j}^{\top}
    \bb)}\right]}^{2},
\end{align*}
where $\mathcal{R}(t) = \{j\mid Y_j(t) = 1\}$ is the index set of subjects in
the risk-set at time $t$.  Define
\begin{align*}
  D_k(\bth^{(i)}) = \frac{1}{4n} \sum_{j=1}^n \int_0^{\infty}
  {\left(\max_{j\in \mathcal{R}(t)} x_{j,k}
  - \min_{j\in \mathcal{R}(t)}x_{j,k}\right)}^2
  \dif N_j(t; \bth^{(i)}).
\end{align*}%
One may verify that the variance of a discrete random variable $X$ with
$\Pr(X = x_{j,k}) = M_j^{(i)} \exp(\bx_{j}^{\top} \bb)/ \sum_{j=1}^n M_j^{(i)}
\exp(\bx_{j}^{\top} \bb)$ is
$- {\partial^2 I(\bb, t \mid \bth^{(i)})}/{\partial \beta_k^2}$.  The variance
will be maximized to
${(\max_{j\in \mathcal{R}(t)} x_{j,k} -
  \min_{j\in \mathcal{R}(t)}x_{j,k})}^2/4$,
if the probability mass of $X$ is equally
distributed on $\max_{j\in \mathcal{R}(t)} x_{j,k}$
and $\min_{j\in \mathcal{R}(t)} x_{j,k}$, which suggests that
\begin{align}\label{eqn:inq:cox}
  -\frac{1}{n}\ddot{\ell}(\bb \mid \bth^{(i)}) \le D_k(\bth^{(i)}),~\forall
  \bb\in \mathbb{R}^p.
\end{align}%
Then we approximated~\eqref{eq:loss-beta-k} with
\begin{align}
  q(\beta_k \mid \tilde{\bb}^{(i)}_{k}, \bth^{(i)}) =
  & - \frac{1}{n} \left[
    \ell(\tilde{\bb}^{(i)}_{k} \mid \bth^{(i)}) +
    \dot{\ell}_k(\tilde{\bb}^{(i)}_{k} \mid
    \bth^{(i)})(\beta_k - {\beta}^{(i)}_k)
    \right]
    + \frac{D_k}{2}{(\beta_k - {\beta}^{(i)}_k)}^{2}\nonumber\\
  & + \lambda_1 \left(\alpha_1 \omega_k \lvert \beta_k \rvert
    + \frac{1 - \alpha_1}{2} \beta_k^2\right),\label{eq:quad-beta-k}
\end{align}
where $\tilde{\bb}^{(i)}_{k} =
{(\tilde{\beta}_1, \ldots, \beta_k^{(i)}, \ldots, \tilde{\beta}_{p})}^{\top}$.
The minimizer of~\eqref{eq:quad-beta-k} can be easily found by the
soft-thresholding rule.  The update of $\tilde{\beta}_k$ is
thus
\begin{align}
  \beta^{(i+1)}_k = \frac{s(D_k(\bth^{(i)})
  {\beta}^{(i)}_k + \frac{1}{n} \dot{\ell}_k(\tilde{\bb}^{(i)}_{k} \mid \bth^{(i)}),
  \lambda_1 \alpha_1 \omega_k)}{D_k(\bth^{(i)}) + \lambda_1(1 - \alpha_1)},
  \label{eq:update-beta-k}
\end{align}
where $s(a, b) = {(\lvert a \rvert - b)}_{+} \mathrm{sign}(a)$ is the
soft-thresholding operator~\citep{donoho1994biometrika}.

\begin{lemma}\label{lem:descent-beta}
  Define the objective function $J(\bb)$ for fixed $\lambda_1$, $\alpha_1$, and
  $\beta_l = \tilde{\beta}_l$, $l \neq k$, to be $J_k(\beta_k)$.  Let
  $\beta_k^{(i)}$ and $\beta_k^{(i+1)}$ denote $\beta_k$ before
  and after the update step in Algorithm 2, respectively,
  $k\in\{1,\ldots,p\}$.  Then
  $J_k(\beta_k^{(i+1)}) \le J_k(\beta_k^{(i)})$.
\end{lemma}

\begin{proof}
  $\forall k \in \{1,\ldots,p\}$, by the Taylor theorem, we have
  \begin{align*}
    & \ell(\beta_k\mid \beta_l = \tilde{\beta}_l,l\neq k, \bth^{(i)}) \\
    =~
    & \ell(\tilde{\bb}^{(i)}_{k} \mid \bth^{(i)}) +
      \dot{\ell}_k(\tilde{\bb}^{(i)}_{k} \mid \bth^{(i)})
      (\beta_k - {\beta}^{(i)}_k) +
      \frac{1}{2}\ddot{\ell}_k(\tilde{\bm{\xi}}_k \mid \bth^{(i)})
      {(\beta_k - {\beta}^{(i)}_k)}^2,
  \end{align*}
  where $\tilde{\bm{\xi}}_{k} = {(\tilde{\beta}_1, \ldots, \xi_k,
    \ldots, \tilde{\beta}_{p})}^{\top}$,
  $\xi_k$ is some real number between $\beta_k$ and $\beta_k^{(i)}$.
  From~\eqref{eqn:inq:cox}, we further have that
  \begin{align*}
    & - \frac{1}{n}\ell(\beta_k\mid \beta_l = \tilde{\beta}_l,l\neq k,
      \bth^{(i)}) \\
    =
    & - \frac{1}{n} \left[
      \ell(\tilde{\bb}^{(i)}_{k} \mid \bth^{(i)}) +
      \dot{\ell}_k(\tilde{\bb}^{(i)}_{k} \mid \bth^{(i)})
      (\beta_k - {\beta}^{(i)}_k)
      \right] -
      \frac{1}{2n}\ddot{\ell}_k(\tilde{\bm{\xi}}_k \mid \bth^{(i)})
      {(\beta_k - {\beta}^{(i)}_k)}^2 \\
    \le
    & - \frac{1}{n}\left[
      \ell(\tilde{\bb}^{(i)}_{k} \mid \bth^{(i)}) +
      \dot{\ell}_k(\tilde{\bb}^{(i)}_{k} \mid \bth^{(i)})
      (\beta_k - {\beta}^{(i)}_k)
      \right] +
      \frac{D_k(\bth^{(i)})}{2}{(\beta_k - {\beta}^{(i)}_k)}^{2},
  \end{align*}
  which suggests that
  $J_k(\beta_k \mid \bth^{(i)}) \le q(\beta_k \mid \tilde{\bb}^{(i)}_{k},
  \bth^{(i)})$, $\forall \beta_k \in \mathbb{R}$.  Notice that
  $J_k({\beta}^{(i)}_k \mid \bth^{(i)}) = q({\beta}^{(i)}_k \mid
  \tilde{\bb}^{(i)}_{k}, \bth^{(i)})$.  So
  $q(\beta_k \mid \tilde{\bb}^{(i)}_{k}, \bth^{(i)})$ majorizes
  $J_k(\beta_k\mid \bth^{(i)})$ at the point ${\beta}^{(i)}_k$.
  For $\beta_k^{(i+1)}$ given in~\eqref{eq:update-beta-k}, we have
  $q({\beta}^{(i+1)}_k \mid \tilde{\bb}^{(i)}_{k}, \bth^{(i)}) \le
    q({\beta}^{(i)}_k \mid \tilde{\bb}^{(i)}_{k}, \bth^{(i)})$.
  Therefore, we can conclude that
  \begin{align*}
    J_k({\beta}^{(i+1)}_k \mid \bth^{(i)}) \le
    q({\beta}^{(i+1)}_k \mid \tilde{\bb}^{(i)}_{k}, \bth^{(i)}) \le
    q({\beta}^{(i)}_k \mid \tilde{\bb}^{(i)}_{k}, \bth^{(i)}) =
    J_k({\beta}^{(i)}_k \mid \bth^{(i)}),
  \end{align*}
  which ensures the monotonic descending of~\eqref{eq:loss-beta-k}.
\end{proof}

We similarly derived the procedure minimizing~\eqref{eqn:cure:reg} based on
the CMD algorithm.
For fixed tuning parameters $\alpha_2$, $\lambda_2$, and
$\gamma_l = \tilde{\gamma}_l$, $l\neq k$, the loss function of $\gamma_k$,
$k\in\{0,\ldots,p\}$ is
\begin{align}
  J_k(\gamma_k \mid \bth^{(i)}) =
  -\frac{1}{n} \ell(\gamma_k \mid
  \gamma_l = \tilde{\gamma}_l,~l\neq k,~\bth^{(i)}) +
  \lambda_2 \bm{1}(k\neq 0) \left( \alpha_2 \nu_k \lvert \gamma_k \rvert +
  \frac{1 - \alpha_2}{2} \gamma_k^2\right).\label{eq:loss-gamma-k}
\end{align}
Let $\dot{\ell}_k(\bg, \gamma_0 \mid \bth^{(i)})$ and
$\ddot{\ell}_k(\bg, \gamma_0 \mid \bth^{(i)})$ denote the first and second
partial derivative of $\ell(\bg \mid \bth^{(i)})$ with respect to
$\gamma_k$, respectively.  We have
\begin{align*}
  \dot{\ell}_0(\bg, \gamma_0 \mid \bth^{(i)})
  & = \sum_{j=1}^{n} M_j^{(i)} - p_j,~
    \ddot{\ell}_0(\bg, \gamma_0 \mid \bth^{(i)})
  = - \sum_{j=1}^{n} p_j (1 - p_j),\\
  \dot{\ell}_k(\bg, \gamma_0 \mid \bth^{(i)})
  & = \sum_{j=1}^{n} (M_j^{(i)} - p_j) \ox_{j,k},~
  \ddot{\ell}_k(\bg, \gamma_0 \mid \bth^{(i)})
  = - \sum_{j=1}^{n} p_j (1 - p_j) \ox_{j,k}^2,~k\neq 0.
\end{align*}
Notice that $p_j(1 - p_j) \le 1 / 4$, $\forall 0 \le p_j \le 1$.
Thus, we have
\begin{align}
  - \frac{1}{n} \ddot{\ell}_k(\bg, \gamma_0 \mid \bth^{(i)}) \le B_k =
  \begin{cases}
    1/4 & k = 0\\
    \sum_{j=1}^{n} \ox_{j,k}^2/(4n) & k\neq 0
  \end{cases},\label{eq:inq:logis}
\end{align}
where $B_k$ does not depend on $\bth^{(i)}$ and thus needs computing only once.
We similarly considered a quadratic approximation of~\eqref{eq:loss-gamma-k} as
follows:
\begin{align}
  q(\gamma_k \mid \tilde{\gamma}_k,\bth^{(i)}) =
  & - \frac{1}{n}
    \left[\ell(\tilde{\gamma}_k \mid \bth^{(i)}) +
    \dot{\ell}_k(\bg, \gamma_0 \mid \bth^{(i)})(\gamma_k - \gamma_k^{(i)}) \right]
    \nonumber\\
  & + B_k {(\gamma_k - \gamma_k^{(i)})}^{2} +
    \bm{1}(k\neq 0)\lambda_2 \left( \alpha_2 \nu_k \lvert \gamma_k \rvert +
    \frac{1 - \alpha_2}{2} \gamma_k^2 \right),\label{eq:quad-gamma-k}
\end{align}
where
$\tilde{\gamma}_k^{(i)} = {(\tilde{\gamma}_0, \ldots, \gamma_k^{(i)}, \ldots,
  \tilde{\gamma}_{p})}^{\top}$.
By soft-thresholding rule, the update step that
minimizes~\eqref{eq:quad-gamma-k} is
\begin{align}
  \gamma^{(i+1)}_k = \frac{s(B_k {\gamma}^{(i)}_k +
  \frac{1}{n} \dot{\ell}_k(\tilde{\bg}^{(i)}_{k} \mid \bth^{(i)}),
  \lambda_2\alpha_2 \nu_k \bm{1}(k\neq 0))}
  {B_k + \lambda_2(1 - \alpha_2)\bm{1}(k\neq 0)},
  \label{eq:update-gamma-k}
\end{align}
where $s(a, b) = {(\lvert a \rvert - b)}_{+} \mathrm{sign}(a)$ is again the
soft-thresholding operator.

\begin{lemma}\label{lem:descent-gamma}
  Define the objective function $J(\bg, \gamma_0)$ for fixed $\lambda_2$,
  $\alpha_2$, and $\gamma_l = \tilde{\gamma}_l$, $l \neq k$, to be
  $J_k(\gamma_k)$.  Let $\gamma_k^{(i)}$ and $\gamma_k^{(i+1)}$ 
  denote $\gamma_k$ before and after the update step in
  Algorithm 3, respectively, $k\in\{0,\ldots,p\}$.  Then
  $J_k(\gamma_k^{(i+1)}) \le J_k(\gamma_k^{(i)})$.
\end{lemma}

\begin{proof}
  $\forall k \in \{0,\ldots,p\}$, by the Taylor theorem, we have
  \begin{align*}
    & \ell(\gamma_k\mid \gamma_l = \tilde{\gamma}_l,l\neq k, \bth^{(i)}) \\
    =~
    & \ell(\tilde{\bg}^{(i)}_{k} \mid \bth^{(i)}) +
      \dot{\ell}_k(\tilde{\bg}^{(i)}_{k} \mid \bth^{(i)})
      (\gamma_k - {\gamma}^{(i)}_k) +
      \frac{1}{2}\ddot{\ell}_k(\tilde{\bm{\xi}}_k \mid \bth^{(i)})
      {(\gamma_k - {\gamma}^{(i)}_k)}^2,
  \end{align*}
  where $\tilde{\bm{\xi}}_{k} = {(\tilde{\gamma}_0, \ldots, \xi_k,
    \ldots, \tilde{\gamma}_{p})}^{\top}$,
  $\xi_k$ is some real number between $\gamma_k$ and $\gamma_k^{(i)}$.
  From~\eqref{eq:inq:logis}, we further have that
  \begin{align*}
    & - \frac{1}{n}\ell(\gamma_k\mid \gamma_l = \tilde{\gamma}_l,l\neq k,
      \bth^{(i)}) \\
    =
    & - \frac{1}{n} \left[
      \ell(\tilde{\bg}^{(i)}_{k} \mid \bth^{(i)}) +
      \dot{\ell}_k(\tilde{\bg}^{(i)}_{k} \mid \bth^{(i)})
      (\gamma_k - {\gamma}^{(i)}_k)
      \right] -
      \frac{1}{2n}\ddot{\ell}_k(\tilde{\bm{\xi}}_k \mid \bth^{(i)})
      {(\gamma_k - {\gamma}^{(i)}_k)}^2 \\
    \le
    & - \frac{1}{n}\left[
      \ell(\tilde{\bg}^{(i)}_{k} \mid \bth^{(i)}) +
      \dot{\ell}_k(\tilde{\bg}^{(i)}_{k} \mid \bth^{(i)})
      (\gamma_k - {\gamma}^{(i)}_k)
      \right] +
      \frac{B_k}{2}{(\gamma_k - {\gamma}^{(i)}_k)}^{2},
  \end{align*}
  which suggests that
  $J_k(\gamma_k \mid \bth^{(i)}) \le q(\gamma_k \mid \tilde{\bg}^{(i)}_{k},
  \bth^{(i)})$, $\forall \gamma_k \in \mathbb{R}$.  Notice that
  $J_k({\gamma}^{(i)}_k \mid \bth^{(i)}) = q({\gamma}^{(i)}_k \mid
  \tilde{\bg}^{(i)}_{k}, \bth^{(i)})$.  So
  $q(\gamma_k \mid \tilde{\bg}^{(i)}_{k}, \bth^{(i)})$ majorizes
  $J_k(\gamma_k\mid \bth^{(i)})$ at the point ${\gamma}^{(i)}_k$.
  For $\gamma_k^{(i+1)}$ given in~\eqref{eq:update-gamma-k}, we have
  $q({\gamma}^{(i+1)}_k \mid \tilde{\bg}^{(i)}_{k}, \bth^{(i)}) \le
    q({\gamma}^{(i)}_k \mid \tilde{\bg}^{(i)}_{k}, \bth^{(i)})$.
  Therefore, we can conclude that
  \begin{align*}
    J({\gamma}^{(i+1)}_k \mid \bth^{(i)}) \le
    q({\gamma}^{(i+1)}_k \mid \tilde{\bg}^{(i)}_{k}, \bth^{(i)}) \le
    q({\gamma}^{(i)}_k \mid \tilde{\bg}^{(i)}_{k}, \bth^{(i)}) =
    J({\gamma}^{(i)}_k \mid \bth^{(i)}),
  \end{align*}
  which ensures the monotonic descending of~\eqref{eq:loss-gamma-k}.
\end{proof}



 \clearpage
\section{Additional Simulation Results}\label{supp:sim}

\subsection{Diagram of Data Generation}
\begin{figure}[htp]
  \centering
  \includegraphics[width=0.8\linewidth]{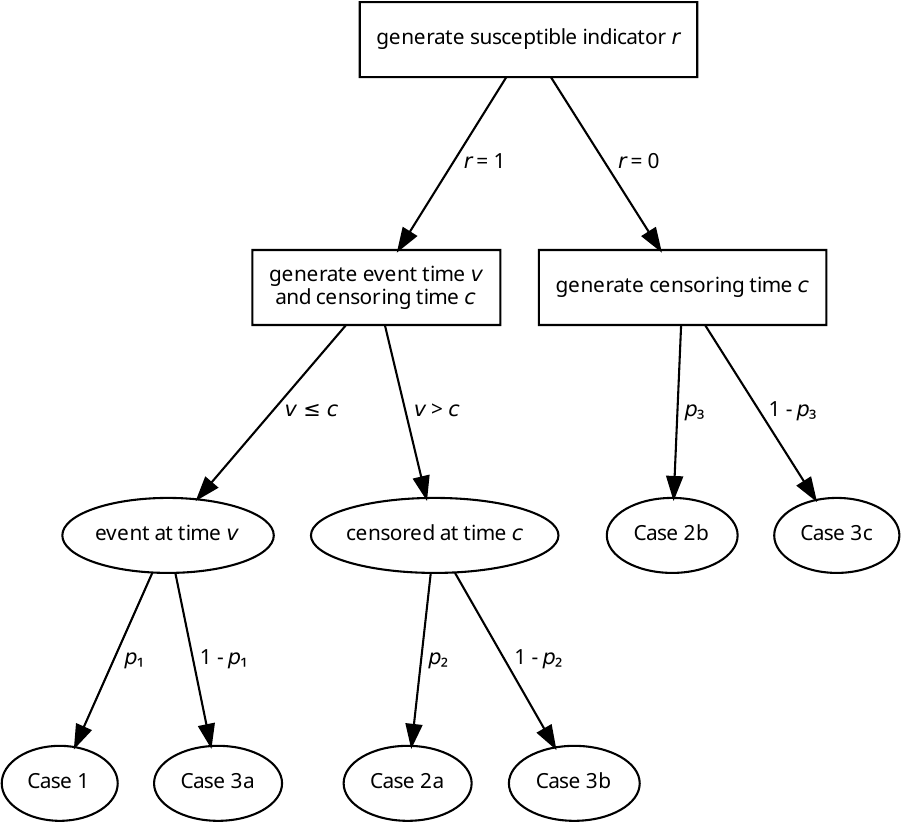}
  \caption[Generation of survival cure data with uncertain events.]
  {Generation of survival data with a cured fraction and uncertain events.}%
  \label{fig:simu:data}
\end{figure}

\clearpage
\subsection{Setting 1}

\begin{table}[h]
  \centering
  \caption{Simulation results of Setting 1: mean area under curve (AUC) in percentage (with standard deviation given in
    parenthesis) of identification of actual events for subjects with
    uncertain events (in Case~3).}%
  \label{tab:simu:mcer}
  \small
  \begin{tabular}{ccccccccc}
    \toprule
           & \multicolumn{8}{c}{Scenario}                                   \\
    \cmidrule(lr){2-9}
           & 1     & 2     & 3     & 4     & 5     & 6     & 7     & 8      \\
    \midrule
    I.Cure & 75.1  & 76.3  & 74.7  & 75.9  & 75.1  & 76.5  & 74.5  & 75.3   \\
           & (4.1) & (6.7) & (4.5) & (5.1) & (5.7) & (5.1) & (4.3) & (17.0) \\
    [1ex]
    Oracle & 88.6  & 89.4  & 87.9  & 88.7  & 88.2  & 88.9  & 87.1  & 88.1   \\
           & (4.5) & (3.1) & (2.9) & (7.2) & (5.4) & (3.2) & (3.6) & (2.2)  \\
    \bottomrule
  \end{tabular}
\end{table}

\begin{table}[htp]
  \centering
  \caption{Simulation results of Setting 1: mean of standard error (SE) estimates measured in scale of $10^{-2}$ for
    the proposed method with the empirical SE given in the parenthesis.}%
  \label{tab:simu:se}
  \small \setlength{\tabcolsep}{5pt}
  \begin{tabular}{cccccccccccc}
    \toprule
             & \multicolumn{5}{c}{Logistic}
             & \multicolumn{5}{c}{Cox}                                                                    \\
    \cmidrule(lr){2-6}\cmidrule(lr){8-12}
    Scenario & $\hat{\mathrm{se}}(\hat{\gamma}_1)$
             & $\hat{\mathrm{se}}(\hat{\gamma}_2)$
             & $\hat{\mathrm{se}}(\hat{\gamma}_3)$
             & $\hat{\mathrm{se}}(\hat{\gamma}_4)$
             & $\hat{\mathrm{se}}(\hat{\gamma}_5)$
             &
             & $\hat{\mathrm{se}}(\hat{\beta}_1)$
             & $\hat{\mathrm{se}}(\hat{\beta}_2)$
             & $\hat{\mathrm{se}}(\hat{\beta}_3)$
             & $\hat{\mathrm{se}}(\hat{\beta}_4)$
             & $\hat{\mathrm{se}}(\hat{\beta}_5)$                  \\
    \midrule
    1        & 41.6   & 39.5   & 42.4   & 32.6   & 35.8
             &        & 12.6   & 12.4   & 13.2   & 14.7   & 15.0   \\
             & (42.4) & (39.8) & (44.4) & (33.7) & (36.3)
             &        & (12.1) & (12.5) & (13.2) & (14.6) & (15.1) \\
    [1ex]
    2        & 27.7   & 26.4   & 27.7   & 22.1   & 23.6
             &        & 13.4   & 13.1   & 14.2   & 15.4   & 15.5   \\
             & (28.7) & (26.7) & (26.8) & (22.2) & (23.6)
             &        & (13.3) & (12.9) & (14.1) & (15.7) & (15.2) \\
    [1ex]
    3        & 36.0   & 33.8   & 37.8   & 28.2   & 31.4
             &        & 11.9   & 11.7   & 12.5   & 13.5   & 13.9   \\
             & (38.2) & (33.9) & (39.5) & (28.6) & (33.1)
             &        & (11.8) & (11.4) & (12.9) & (13.6) & (14.3) \\
    [1ex]
    4        & 23.4   & 22.6   & 24.0   & 19.3   & 20.4
             &        & 12.5   & 12.1   & 13.3   & 13.7   & 14.0   \\
             & (24.7) & (22.5) & (24.3) & (18.8) & (20.5)
             &        & (12.6) & (11.8) & (13.2) & (14.1) & (13.8) \\
    [1ex]
    5        & 39.6   & 37.0   & 40.8   & 31.2   & 34.3
             &        & 12.1   & 11.9   & 12.9   & 13.7   & 13.8   \\
             & (43.2) & (38.2) & (44.6) & (35.0) & (36.0)
             &        & (11.9) & (11.9) & (13.1) & (14.3) & (14.6) \\
    [1ex]
    6        & 25.8   & 24.7   & 26.3   & 21.4   & 22.7
             &        & 13.0   & 12.3   & 13.5   & 13.9   & 14.2   \\
             & (24.9) & (25.5) & (26.2) & (21.7) & (23.1)
             &        & (12.4) & (12.4) & (13.1) & (13.9) & (13.7) \\
    [1ex]
    7        & 31.4   & 29.5   & 33.4   & 25.4   & 27.3
             &        & 11.0   & 10.6   & 11.8   & 12.0   & 12.3   \\
             & (32.8) & (30.2) & (36.4) & (27.0) & (28.1)
             &        & (11.7) & (10.8) & (12.2) & (12.3) & (12.7) \\
    [1ex]
    8        & 20.6   & 19.8   & 21.5   & 17.5   & 18.4
             &        & 11.4   & 10.8   & 12.2   & 11.9   & 12.2   \\
             & (21.1) & (19.9) & (21.2) & (17.3) & (18.6)
             &        & (11.4) & (10.8) & (12.7) & (11.2) & (12.3) \\
    \bottomrule
  \end{tabular}
\end{table}

\clearpage
\subsection{Setting 2}

For each
simulated dataset (the training set) that used for variable selection, we
randomly generated an additional testing set independently from the training set
for evaluating of the prediction performance.  The size of the testing set was
set to be equal with the training set.  We used the susceptible probability
$p_j$ defined in~\eqref{eqn:logis} of the main paper to predict the cure status and evaluated the
prediction on the cure status by the regular AUC for binary outcomes.  For
survival outcomes, we computed Harrell's Concordance index among susceptible
subjects in the testing set based on the estimated risk scores
$\bx^{\top}\hat{\bb}$ from the fitted model.
The mean and standard deviation of the out-of-sample AUC for the prediction on
cure status and out-of-sample C-index are summarized in
Table~\ref{tab:simu:pred}.

\begin{table}[tbp]
  \centering
  \caption{Simulation results of Setting 2: comparison on prediction performance on cure status and survival
    outcomes, respectively, through mean of the area under curve (AUC) and the
    weighted C-index (in percentage) computed on the testing sets (with the standard
    deviation given in parenthesis).}\label{tab:simu:pred}
  \small \setlength{\tabcolsep}{3pt}
  \begin{tabular}{cccccccccccc}
    \toprule
             & \multicolumn{5}{c}{AUC}
             &        & \multicolumn{5}{c}{C-index}                                                    \\
    \cmidrule(lr){2-6}\cmidrule(lr){8-12}
    Scenario & O.Cure & I.Cure & Cure1  & Cure2  & Cure3  &  & O.Cure & I.Cure & Cure1 & Cure2 & Cure3 \\
    \midrule
    1        & 58.5   & 60.6   & 59.2   & 58.1   & 53.3   &  & 84.7   & 84.2   & 83.5  & 83.4  & 83.5  \\
             & (10.7) & (9.9)  & (10.3) & (9.8)  & (7.4)  &  & (3.5)  & (4.7)  & (6.0) & (6.7) & (5.3) \\
    [1ex]
    2        & 75.9   & 74.1   & 74.8   & 74.8   & 62.4   &  & 83.5   & 82.5   & 82.3  & 82.0  & 81.1  \\
             & (10.8) & (10.5) & (10.4) & (10.0) & (13.7) &  & (2.7)  & (4.8)  & (5.0) & (5.5) & (4.8) \\
    [1ex]
    3        & 61.3   & 64.8   & 61.5   & 59.9   & 57.1   &  & 84.8   & 84.6   & 83.3  & 82.5  & 84.6  \\
             & (12.2) & (10.8) & (11.0) & (10.4) & (11.0) &  & (1.7)  & (2.7)  & (5.4) & (7.1) & (1.8) \\
    [1ex]
    4        & 78.5   & 78.7   & 77.4   & 75.8   & 72.9   &  & 83.5   & 83.4   & 82.2  & 81.6  & 82.6  \\
             & (9.6)  & (7.7)  & (8.4)  & (9.6)  & (13.5) &  & (1.9)  & (1.9)  & (3.9) & (4.8) & (2.5) \\
    [1ex]
    5        & 60.5   & 62.3   & 60.4   & 59.1   & 51.6   &  & 84.8   & 84.0   & 83.6  & 83.1  & 83.8  \\
             & (12.0) & (9.8)  & (10.7) & (10.1) & (5.6)  &  & (2.4)  & (4.3)  & (5.0) & (5.8) & (3.4) \\
   [1ex]
    6        & 78.4   & 76.0   & 76.9   & 75.8   & 57.0   &  & 83.6   & 82.1   & 82.4  & 81.6  & 80.5  \\
             & (9.5)  & (8.6)  & (9.0)  & (9.5)  & (12.0) &  & (2.0)  & (5.0)  & (3.9) & (6.0) & (3.3) \\
    [1ex]
    7        & 67.1   & 71.0   & 66.2   & 61.9   & 59.5   &  & 84.4   & 84.3   & 82.6  & 81.7  & 84.2  \\
             & (13.9) & (9.8)  & (11.6) & (11.1) & (13.0) &  & (1.4)  & (1.8)  & (5.1) & (6.6) & (1.4) \\
    [1ex]
    8        & 81.4   & 81.3   & 80.0   & 78.7   & 72.7   &  & 83.2   & 83.0   & 82.0  & 81.3  & 81.7  \\
             & (6.8)  & (3.6)  & (5.4)  & (6.2)  & (14.5) &  & (1.7)  & (1.6)  & (2.5) & (3.9) & (2.6) \\
    \bottomrule
  \end{tabular}
\end{table}


In addition to using unit weights, we considered the adaptive weights
$\omega_k = \lvert \tilde{\beta}_k^{-1} \rvert$, and
$\nu_k = \lvert \tilde{\gamma}_k^{-1} \rvert$, where $\tilde{\bb}$ and
$\tilde{\bg}$ were the non-regularized estimates, $k\in\{1,\ldots,p\}$.  The
evaluation of variable selection performance is summarized in
Table~\ref{tab:simu:var:adaptive}.  The mean and standard deviation of the
out-of-sample AUC for the prediction on cure status and out-of-sample C-index
are summarized in Table~\ref{tab:simu:pred:adaptive}.

\begin{table}[btp]
  \centering
  \caption{Additional simulation results of Setting 2: comparison on variable selection performance for the logistic model
    part and Cox model part through mean of true positive rate
    (TPR) and false positive rate (FPR) in percentage (with the standard
    deviation given in parenthesis).}\label{tab:simu:var:adaptive}
  \small \setlength{\tabcolsep}{3pt}
  \begin{tabular}{ccccccccccccc}
    \toprule
             &         & \multicolumn{5}{c}{Logistic}
             &         & \multicolumn{5}{c}{Cox}                                                                    \\
    \cmidrule(lr){3-7}\cmidrule(lr){9-13}
    Scenario & Measure & O.Cure & I.Cure & Cure1  & Cure2  & Cure3  &  & O.Cure & I.Cure & Cure1  & Cure2  & Cure3  \\
    \midrule
    1        & TPR     & 29.1   & 32.4   & 26.0   & 24.5   & 21.4   &  & 83.0   & 84.4   & 78.7   & 76.5   & 61.7   \\
             &         & (22.3) & (26.7) & (20.7) & (18.1) & (24.6) &  & (18.2) & (20.7) & (13.3) & (21.7) & (18.9) \\
             & FPR     & 2.01   & 2.86   & 2.03   & 1.98   & 1.53   &  & 2.78   & 3.73   & 2.87   & 2.74   & 2.56   \\
             &         & (2.24) & (2.97) & (3.38) & (1.48) & (2.58) &  & (2.76) & (3.34) & (3.75) & (2.24) & (2.85) \\
    [1ex]
    2        & TPR     & 57.4   & 57.5   & 50.4   & 48.5   & 36.3   &  & 85.7   & 87.2   & 80.8   & 79.1   & 67.8   \\
             &         & (22.8) & (23.5) & (24.6) & (21.4) & (22.9) &  & (19.0) & (21.3) & (20.3) & (13.5) & (23.4) \\
             & FPR     & 2.43   & 3.55   & 2.65   & 2.74   & 1.46   &  & 2.79   & 4.08   & 3.31   & 3.36   & 2.85   \\
             &         & (2.94) & (1.91) & (2.86) & (2.41) & (2.54) &  & (3.63) & (2.65) & (3.03) & (2.22) & (3.07) \\
    [1ex]
    3        & TPR     & 36.1   & 42.5   & 28.8   & 25.0   & 29.9   &  & 88.9   & 89.9   & 80.3   & 76.4   & 83.0   \\
             &         & (21.5) & (25.1) & (26.5) & (21.7) & (22.3) &  & (20.7) & (13.8) & (23.2) & (14.7) & (12.7) \\
             & FPR     & 1.79   & 3.11   & 1.98   & 1.99   & 1.32   &  & 2.45   & 3.42   & 2.96   & 2.89   & 2.46   \\
             &         & (2.34) & (2.12) & (2.97) & (4.01) & (1.59) &  & (2.76) & (2.24) & (2.86) & (4.46) & (1.82) \\
    [1ex]
    4        & TPR     & 70.1   & 70.5   & 56.7   & 53.9   & 56.3   &  & 89.2   & 91.4   & 79.8   & 75.9   & 83.4   \\
             &         & (21.9) & (24.0) & (24.1) & (25.2) & (15.0) &  & (22.4) & (13.7) & (16.5) & (20.5) & (9.5)  \\
             & FPR     & 2.48   & 3.75   & 2.85   & 2.83   & 1.83   &  & 2.11   & 3.63   & 3.15   & 3.02   & 2.30   \\
             &         & (2.31) & (2.80) & (2.06) & (2.76) & (2.06) &  & (2.71) & (3.34) & (2.26) & (3.08) & (1.87) \\
    [1ex]
    5        & TPR     & 35.0   & 40.6   & 28.0   & 24.3   & 20.3   &  & 87.6   & 88.5   & 79.7   & 76.7   & 63.6   \\
             &         & (19.0) & (23.7) & (24.0) & (25.5) & (14.9) &  & (24.1) & (18.8) & (14.8) & (21.4) & (8.3)  \\
             & FPR     & 1.85   & 3.73   & 2.10   & 1.99   & 1.14   &  & 2.41   & 4.17   & 2.83   & 2.92   & 2.24   \\
             &         & (1.75) & (2.37) & (2.10) & (2.88) & (3.84) &  & (2.46) & (3.07) & (2.43) & (2.91) & (4.59) \\
    [1ex]
    6        & TPR     & 67.9   & 68.4   & 55.8   & 52.5   & 32.7   &  & 89.0   & 91.2   & 79.7   & 76.2   & 72.6   \\
             &         & (24.9) & (22.6) & (23.2) & (22.3) & (23.9) &  & (16.0) & (21.7) & (14.4) & (19.7) & (19.0) \\
             & FPR     & 2.41   & 4.49   & 2.81   & 2.84   & 1.05   &  & 2.29   & 4.73   & 3.10   & 3.04   & 2.63   \\
             &         & (2.60) & (2.46) & (3.20) & (1.49) & (2.79) &  & (2.55) & (2.87) & (4.08) & (2.49) & (3.32) \\
    [1ex]
    7        & TPR     & 53.8   & 60.2   & 35.0   & 28.6   & 37.8   &  & 93.1   & 95.4   & 80.1   & 73.6   & 89.1   \\
             &         & (25.1) & (22.2) & (23.9) & (24.0) & (24.4) &  & (16.9) & (17.4) & (18.7) & (10.4) & (22.1) \\
             & FPR     & 1.73   & 4.35   & 2.22   & 2.07   & 1.07   &  & 1.64   & 3.69   & 2.87   & 2.99   & 1.79   \\
             &         & (3.50) & (1.72) & (2.46) & (2.06) & (2.84) &  & (4.07) & (2.25) & (2.84) & (1.87) & (2.98) \\
    [1ex]
    8        & TPR     & 82.7   & 83.8   & 64.1   & 59.1   & 62.5   &  & 94.0   & 96.6   & 80.7   & 73.9   & 90.3   \\
             &         & (25.8) & (21.0) & (22.7) & (22.1) & (21.3) &  & (19.3) & (13.7) & (21.6) & (9.0)  & (12.6) \\
             & FPR     & 1.93   & 4.50   & 2.87   & 2.88   & 1.34   &  & 1.55   & 4.17   & 3.11   & 3.00   & 1.84   \\
             &         & (2.91) & (2.53) & (2.43) & (3.77) & (1.62) &  & (3.28) & (2.11) & (2.76) & (4.03) & (1.96) \\
    \bottomrule
  \end{tabular}
\end{table}

\begin{table}[tbp]
  \centering
  \caption{Additional simulation results of Setting 2: comparison on prediction performance on cure status and survival
    outcomes, respectively, through mean of the area under curve (AUC) and the
    weighted C-index (in percentage) computed on the testing sets (with the standard
    deviation given in parenthesis).}\label{tab:simu:pred:adaptive}
  \small \setlength{\tabcolsep}{3pt}
  \begin{tabular}{cccccccccccc}
    \toprule
             & \multicolumn{5}{c}{AUC}
             &        & \multicolumn{5}{c}{C-index}                                                  \\
    \cmidrule(lr){2-6}\cmidrule(lr){8-12}
    Scenario & O.Cure & I.Cure & Cure1 & Cure2 & Cure3  &  & O.Cure & I.Cure & Cure1 & Cure2 & Cure3 \\
    \midrule
    1        & 64.8   & 65.0   & 63.0  & 62.0  & 62.0   &  & 83.8   & 83.5   & 83.0  & 82.5  & 79.7  \\
             & (8.9)  & (8.1)  & (9.0) & (9.1) & (8.6)  &  & (3.3)  & (4.3)  & (4.6) & (5.3) & (7.0) \\
    [1ex]
    2        & 75.8   & 74.8   & 73.3  & 72.3  & 69.8   &  & 82.8   & 82.4   & 81.7  & 81.2  & 79.4  \\
             & (7.4)  & (7.2)  & (8.7) & (9.2) & (9.4)  &  & (2.8)  & (3.3)  & (4.0) & (4.8) & (5.3) \\
    [1ex]
    3        & 68.2   & 68.8   & 64.2  & 62.3  & 66.6   &  & 84.3   & 84.1   & 82.9  & 82.0  & 83.6  \\
             & (9.4)  & (7.5)  & (9.6) & (9.5) & (9.5)  &  & (2.0)  & (2.1)  & (3.5) & (5.1) & (3.0) \\
    [1ex]
    4        & 79.4   & 78.9   & 75.5  & 74.1  & 76.7   &  & 83.1   & 82.9   & 81.1  & 80.2  & 82.2  \\
             & (4.6)  & (4.7)  & (7.2) & (8.1) & (6.6)  &  & (2.0)  & (2.3)  & (4.2) & (5.4) & (2.5) \\
    [1ex]
    5        & 67.7   & 67.2   & 63.5  & 61.7  & 62.4   &  & 84.4   & 84.0   & 83.1  & 82.3  & 80.6  \\
             & (9.1)  & (7.2)  & (9.7) & (9.5) & (9.0)  &  & (2.1)  & (2.5)  & (3.1) & (4.7) & (5.4) \\
   [1ex]
    6        & 78.8   & 77.1   & 75.0  & 73.7  & 69.4   &  & 83.2   & 82.5   & 81.2  & 80.6  & 79.7  \\
             & (5.2)  & (5.1)  & (7.4) & (7.9) & (9.9)  &  & (2.0)  & (3.0)  & (4.2) & (4.9) & (4.1) \\
    [1ex]
    7        & 75.2   & 73.8   & 67.0  & 63.9  & 71.2   &  & 84.4   & 84.3   & 82.5  & 81.0  & 84.0  \\
             & (6.9)  & (5.4)  & (9.6) & (9.5) & (8.5)  &  & (1.5)  & (1.5)  & (3.2) & (5.3) & (1.7) \\
    [1ex]
    8        & 82.2   & 81.3   & 77.4  & 76.0  & 79.1   &  & 83.1   & 82.8   & 80.5  & 79.4  & 82.4  \\
             & (1.5)  & (2.2)  & (6.6) & (6.9) & (4.6)  &  & (1.5)  & (1.6)  & (3.8) & (4.8) & (1.7) \\
    \bottomrule
  \end{tabular}
\end{table}

\clearpage
\section{Additional Results for the Suicide Risk Study}


We refitted the selected model without regularization and performed bootstrap to
obtain 95\% confidence intervals of the coefficient estimates. According to \citet{zhao2017defense}, such a naive two-step procedure could
still yield asymptotically valid inference under certain
conditions. Specifically, the SE estimates were obtained from 1,000 bootstrap
samples and the confidence interval were estimated based on asymptotic normality of the coefficient estimates. The results
are shown in Table~\ref{tab:relax:joint} of Supplementary Materials. Most of the
predictors were significant at $\alpha=0.05$ significance level, except ICD-9 298, 300, and V62 in the
incidence part.

\begin{table}[htp]
  \caption{Exponentiated coefficient estimates (hazard ratios or odd ratios) of
    the refitted I.Cure model with 95\% confidence intervals.}\label{tab:relax:joint}
\centering
\begin{tabular}{crrr}
  \toprule
  ICD-9 & HR/OR & Lower & Upper \\
  \midrule
  [0.5ex]
  \multicolumn{4}{c}{Survival (Latency) Part}\\
  [0.5ex]
  304 & 1.46 & 1.11 & 1.93 \\
  V62 & 1.40 & 1.09 & 1.80 \\
  [1ex]
  \multicolumn{4}{c}{Incidence Part}\\
  [0.5ex]
  296 & 1.95 & 1.57 & 2.44 \\
  298 & 1.19 & 0.92 & 1.56 \\
  300 & 1.22 & 0.99 & 1.50 \\
  301 & 1.64 & 1.21 & 2.21 \\
  304 & 1.55 & 1.13 & 2.14 \\
  312 & 1.45 & 1.00 & 2.10 \\
  313 & 1.92 & 1.22 & 3.01 \\
  319 & 3.85 & 1.15 & 12.96 \\
  564 & 1.87 & 1.21 & 2.90 \\
  V62 & 1.22 & 0.95 & 1.57 \\
   \bottomrule
\end{tabular}
\end{table}


\clearpage
\bibliographystyle{asa}
\bibliography{cure-mci.bib}

\end{document}